\documentclass[12pt,reqno]{amsart}

\PassOptionsToPackage{dvipsnames,svgnames}{xcolor}
\usepackage{xcolor}

\usepackage{amsmath}
\usepackage{amssymb}
\usepackage{amsfonts}
\usepackage{mathscinet}
\usepackage{diagbox}
\usepackage{booktabs}

\usepackage{graphicx}
\usepackage{tikz}
\usetikzlibrary{arrows.meta}
\tikzset{
  mynode/.style={draw, circle, minimum size=1cm, inner sep=0pt},
  myarrow/.style={very thick},
}
\usepackage{multirow}
\usepackage{float}
\usepackage{subcaption}

\usepackage{algorithm}
\usepackage{algorithmic}
\usepackage[left=2.9cm,right=2.9cm,top=3cm,bottom=3cm]{geometry}
\usepackage{enumerate}
\usepackage{verbatim}
\usepackage[title]{appendix}

\newcommand{\N}{\mathbb N}

\usepackage{amsthm}

\newtheorem{theorem}{Theorem}
\newtheorem{lemma}{Lemma}

\newtheorem{proposition}{Proposition}
\newtheorem{assumption}{Assumption}

\newtheorem{corollary}{Corollary}

\theoremstyle{definition} 
\newtheorem{rem}{Remark}
\newtheorem{definition}{Definition}
\date{\today}

\title{Strategic Partitioning and Manipulability in Two-Round Elections}

\author{Emilio {\bf De Santis}}
\address{Universit\`a di Roma La Sapienza, Dipartimento di Matematica Guido Castelnuovo, 
Piazzale Aldo Moro, 5, 00185, Rome, Italy -- ORCID: 0000-0001-9563-7685}
\email{desantis@mat.uniroma1.it}
\author{Antonio {\bf Di Crescenzo}}
\address{Universit\`a degli Studi di Salerno, Dipartimento di Matematica,
 Via Giovanni Paolo II, 132, 84084, Fisciano (SA), Italy -- ORCID: 0000-0003-4751-7341}
\email{adicrescenzo@unisa.it}

\author{Verdiana {\bf Mustaro}}
\address{Universit\`a degli Studi di Salerno, Dipartimento di Matematica,
 Via Giovanni Paolo II, 132, 84084, Fisciano (SA), Italy -- ORCID: 0000-0003-4583-2612}
\email{vmustaro@unisa.it}

\begin{document}
\begin{abstract} 
We consider a two-round election model involving $m$ voters and $n$ candidates. Each voter is endowed with a strict preference list ranking the candidates. In the first round, the candidates are partitioned into two subsets, $A$ and $B$, and voters select their preferred candidate from each. Provided there are no ties, the two respective winners advance to a second round, where voters choose between them according to their initial preference lists. We analyze this scenario using a probabilistic framework based on a spatial voting model with cyclically constructed preference lists and uniformly distributed ideal points. Our objective is to determine the optimal initial partition of $A$ and $B$ that maximizes a target candidate's probability of winning. We analytically evaluate this success probability and derive its asymptotic behavior as the number of candidates $n \to \infty$. A key finding is that the asymptotically optimal relative width of the main discrete cluster converges precisely to one-fifth of the total number of candidates. Finally, we provide computational results and confidence intervals derived from simulation algorithms that validate the analytical framework. Specifically, we demonstrate that the probability of the universal victory event rapidly approaches $1$ as the electorate size increases.


\end{abstract}

\maketitle

% ======================================================
\section{Introduction} 
% ======================================================

Decision-making mechanisms concerning voting procedures play a fundamental role in political science, economics, and operations research. Within this context, a recurring theme is the analysis of voting rules under uncertainty, especially when voters' preferences are subject to random fluctuations. While many voting systems are designed with the aim of ensuring fair outcomes, they often exhibit vulnerabilities to structural manipulation. In this paper, we analyze a widely adopted two-round voting scheme under the lens of \emph{agenda control}: we investigate how the initial partition of candidates into subsets can predetermine the final winner. 

In our model, an election involves $m$ voters and $n$ candidates. Each voter is endowed with a strict preference list, classically represented as a permutation of all candidates. The voting mechanism operates as follows: whenever voters are required to choose from a given subset of candidates, they assign their vote to the candidate in that subset who appears highest in their individual preference list. Initially, the entire set of candidates is divided into two mutually exclusive groups, $A$ and $B$. In the first round, applying the aforementioned rule, the voters select a favorite candidate from $A$ and one from $B$. In the second round, all voters express a preference between the two winners of the first round. 

It is important to emphasize a strict rule regarding ties in our model: if a tie occurs among the top candidates in either the primary elections (within $A$ or $B$) or in the final ballot, the election is considered void and no candidate is declared the winner. Consequently, the sum of the winning probabilities of all individual candidates is strictly less than $1$, reflecting the strictly positive probability of such unresolved draws.

In a deterministic setting, the power of the agenda setter is highly dependent on the specific preference profiles of the voters. While trivial scenarios (such as unanimous voters) preclude any manipulation, there exist strategically vulnerable profiles where the agenda setter can completely dictate the outcome, steering the victory toward any desired candidate simply by tailoring the subsets $A$ and $B$. 

The vulnerability of aggregation systems to such strategic manipulation is a central theme in social choice theory. In the context of knockout tournaments, foundational works and the extensive literature on the \emph{Tournament Fixing Problem} have demonstrated how the final outcome can be decisively manipulated simply by acting on the initial seeding. The question of how much preference profiles can be structured to generate arbitrary rankings might seem analytically intractable. However, recent studies have shown that the problem can be formally addressed by exploiting an isomorphism with load-sharing models from reliability theory \cite{DSSp2023, DeSantisSpizzichino2023_DEF}. In these works, it was constructively proven that it is always possible to generate a voter population capable of realizing any arbitrary family of rankings. Building on this constructive approach, a recent paper \cite{desantisdebaets} quantified this fragility in single-elimination tournaments: for a tournament with $2^n$ candidates ($n \ge 3$), a specific profile of merely $4n-3$ voters is sufficient to guarantee that \emph{any} candidate can win, simply by choosing a proper initial board order. Furthermore, in the same work, a probabilistic scenario was explored, demonstrating that universal manipulability could be achieved with a number of voters proportional to $(\log n)^3$.

While our previous work focused on the sequential nature of single-elimination tournaments (where the manipulator can exploit the depth of the tournament tree), the present paper tackles the intrinsically different and ``flatter'' dynamics of a two-round voting system. In the tournament setting, manipulability is achieved with a relatively small electorate (polylogarithmic in $n$). In contrast, for the two-round system analyzed here, we establish a different asymptotic guarantee: the probability of universal manipulation converges to $1$ as the number of voters $m \to \infty$. 

To understand the vulnerability of this two-round system, we frame the manipulation as an adversarial mechanism design problem. First, the agenda setter (acting as the manipulator) designs an ``urn'' containing a specific distribution of preference lists. Since the adversary will announce the target candidate \emph{after} the urn is fixed, the manipulator cannot simply bias the urn toward a specific individual; doing so would allow the adversary to trivially win the challenge by choosing a disadvantaged candidate. Once the urn is established, the adversary randomly draws an electorate of $m$ voters from it (with replacement). 

At this stage, the agenda setter operates in an \emph{open-loop} control regime: they are entirely blind to the actual preference lists drawn by the $m$ voters, knowing only their initial urn design. The adversarial challenge can then unfold at two distinct levels. In the standard challenge, the adversary reveals a single target candidate, and the manipulator must ensure their victory by choosing an appropriate initial partition of candidates into subsets $A$ and $B$. We denote this \emph{individual winning guarantee} as event $F_1$. 

In the second, more extreme version of the challenge, the adversary is allowed to repeatedly change the target candidate \emph{after} the single electorate has been drawn. To win, the manipulator must be able to guarantee the victory of \emph{any} sequentially requested candidate using that exact same electorate, adapting the outcome solely by modifying the initial subsets $A$ and $B$. We denote this \emph{universal victory event} as $F^{(2)}$. The fundamental research question is: does there exist a symmetric composition of the urn that grants the manipulator this open-loop control, guaranteeing both the individual and the universal victory, with high probability, regardless of the specific electorate realization? 

In this paper, we answer affirmatively by demonstrating that an urn consisting of cyclically constructed preference lists provides exactly this level of control. It is crucial to emphasize that these strictly cyclic preferences are not proposed as a realistic sociological or behavioral model of the electorate. Rather, they represent the mathematical solution to the adversary's challenge. They act as an extreme theoretical environment---an upper bound for manipulability---that effectively masks pairwise defeats and unequivocally proves the structural fragility of the voting system under blind manipulation.

Interestingly, we show that increasing the number of candidates $n$ does not hinder the agenda setter; rather, it facilitates the manipulation. A larger candidate pool allows the manipulator to exploit the dispersion of votes among the target's opponents, making it easier to isolate and eliminate dangerous competitors in the first round. 

By employing a probabilistic analysis, we analytically show that $\mathsf{P}(F^{(2)}) \to 1$ as $m \to \infty$, irrespective of $n$. Beyond this asymptotic limit, we demonstrate that the practical size of the electorate required to obtain such a universal guarantee is surprisingly low. To achieve this manipulation, we derive a key geometric property: the asymptotically optimal relative width of the main clusters in the partitions $A$ and $B$ converges precisely to $1/5$ of the total number of candidates. Our numerical simulations confirm that universal manipulability is reached rapidly: for an electorate of just $m = 51$ voters (and $n=30$ candidates), the probability of the universal event $F^{(2)}$ already exceeds $91\%$. This confirms the severe structural vulnerability of the system, showing that near-total control can be established even in small-scale scenarios.

To study the regime where the number of candidates is large relative to the electorate, we introduce a \emph{continuous limit model}. By mapping the discrete cyclic preferences onto the continuous domain $[0,1)$, voters are treated as independent uniform random variables. We show that, provided the number of voters grows as $m = o(\sqrt{n})$, the discrete winning probabilities converge to this continuous limit. This formulation yields closed-form polynomial expressions for the victory probabilities and justifies the convergence of the discrete optimal cluster width to its continuous counterpart. Finally, this framework allows us to extend our asymptotic results to the universal victory event.

% === PARAGRAFO CON LE CITAZIONI RECUPERATE ===
The exploration of such paradoxes and probabilistic bounds is well-rooted in the broader literature. The choice of a two-round voting system is well established; see, for example, \cite{CM2020}. A broader perspective on voting paradoxes and group coherence can be found in \cite{GL2012}. In \cite{Hazla_etal2020}, the authors discuss the phenomenon of intransitivity in different models. Intransitivity in dice tournaments is also assessed in \cite{SaSa2024}, which proves a series of asymptotic results about their distribution. A subsequent study \cite{SaSa25} presents a random majority dynamics scenario to calculate the asymptotical winning probability of one opinion over another. The issue of predetermining victory through structural choices is also tackled in \cite{KSV17}, which proves sufficient conditions to assure that each player can be favored to win by means of a proper initial seeding.
% ==============================================

The plan of the paper is as follows. In Section \ref{sect:prelim}, we introduce the voting system and the related notation. Section \ref{main} is devoted to the main results, with special emphasis on the asymptotic optimal ratio disclosed in Theorem \ref{Bound}. Section \ref{continuo} formalizes the continuous limit and the universal victory event. Finally, in Section \ref{computational} we illustrate the effectiveness of the analysis through computational results, leading to numerical estimates of the optimal cluster length and demonstrating the rapid onset of universal manipulability.

% ==============================
\section{Preliminary results}\label{sect:prelim}
% ==============================
  In order to describe the voting system, 
  we give few preliminary definitions. For simplicity, the set of all possible candidates will be denoted as  $[n]=\{1,2,\ldots,n\}$ and the set of voters as 
  $\mathcal{V}=\{v_1,v_2,\ldots,v_m\}$, where $m,n \in \mathbb{N}$. For any $j\in [m]$, we define the  \emph{preference list $ {\bf r}_j=(r_{j,1},r_{j,2},\ldots,r_{j,n})$
of voter $v_j$} as a permutation  of the elements in $[n]$. 
In the preference list $ {\bf r}_j $,  all the candidates are listed, from left to right, in order of preference of voter $v_j$ (no ties are allowed). 
 Every voter expresses one vote in any election. 
 We consider  elections where a set $A \subseteq [n]$ of candidates \emph{participates}, i.e. only candidates belonging to set $A$ can be voted for by voters.

Therefore, for any $j \in [m]$ voter $v_j$ prefers the candidate $r_{j,h}$ over the candidate $r_{j,k}$, for all $h<k$, where $h,k\in [n]$. 
Hence, if a voting round only contains a set of candidates  $A\subseteq [n]$ then 
the preferred candidate of voter $v_j$ is $r_{j, \ell}$, 
where $\ell =\min\{i\in [n]:r_{j,i}\in A\}$.  Each voter will always cast their vote in each election, expressing a preference. 
  \begin{definition}      In a vote with fixed sets of candidates $[n]$ and voters $\mathcal{V}=\{v_1,v_2,\ldots,v_m\}$ we define the voting profile $R$ (or preference matrix) as an $m\times n$ matrix such that
      \begin{equation} \label{voting_profile}
          R=||r_{ji}||=\left[\begin{matrix}
          {\bf r}_1 \\
          \vdots \\
          {\bf r}_m
      \end{matrix}\right],
      \end{equation}
      where ${\bf r}_j$ represents the preference list of voter $v_j$.
  \end{definition}
    In the following, we will focus on specific choices for the voting preferences. Indeed, we shall restrict our investigation to voting preferences in which the candidates are oriented clockwise.
\begin{definition}\label{def:Vc}
 Given $m$ voters and $n$ candidates, for any $k \in [n]$, the preference list
    $$
(k,k+1,\ldots,n-1,n,1,\ldots, k-1)
$$
     is said to be oriented clockwise.  
      Let ${\mathcal S}_n$ denote the set of the $n$ clockwise oriented preference lists. 
We define $O_{m,n}$ as the set of voting profiles in which each row of the voting profile is a clockwise oriented preference list, according to ${\mathcal S}_n$.   
  \end{definition}
  For $R \in O_{m,n}$, it will be useful to consider the graph 
 $G_n=(V_n,E_n)$, in which  (see Figure \ref{grafo})
  $$V_n:=[n], \qquad E_n:=\{\{i,i+1\}_{1\leq i\leq n-1},\{n,1\}\}.$$
 For each given subset $A\subseteq [n]$, we say that an edge $\{i,i'\} \in E_
 n$ is \emph{$A$-open} if and only if $i,i'\in A$ or $i,i'\in A^c$, where $$A^c=[n]\setminus A.$$
 Furthermore, we say that two nodes $i,i'\in [n]$ are connected if and only if there exists a path of $A$-open edges from $i$ to $i'$. Thus, we partition all nodes into maximal connected components, which are called \emph{clusters}. Given $A \subseteq [n]$, we denote by $\mathcal{C}_{n,A}$ the collection of clusters of $G_n$ generated by $A$. We notice that $\mathcal{C}_{n,A}=\mathcal{C}_{n,A^c}$.
  For each $i\in [n]$, we also define 
  \begin{equation} \label{def_cluster}
      C_{A^*}(i):=\begin{cases}
      \emptyset & {\rm if}\; i \notin A^*, \\
      C\in\mathcal{C}_{n,A}:i\in C &{\rm if}\; i \in A^{*},
  \end{cases}
  \end{equation}
  where $A^*\in \{A,A^c\}$.
    \begin{figure}[!ht]
    \centering
\begin{tikzpicture}[scale=1.4]
  \foreach \i/\name in {
     0/1, 1/2, 2/3, 3/4} {
    \node[draw=red, circle, minimum size=0.8cm, inner sep=0pt] (n\name) at ({-360/14*\i + 90}:2cm) {\name};
  }

    \foreach \i/\name in {10/11,
     11/12, 12/13,13/14} {
    \node[draw=blue, circle, minimum size=0.8cm, inner sep=0pt] (n\name) at ({-360/14*\i + 90}:2cm) {\name};
  }

    \foreach \i/\name in {4/5} {
    \node[draw=blue, circle, minimum size=0.8cm, inner sep=0pt] (n\name) at ({-360/14*\i + 90}:2cm) {\name};
  }

  \foreach \i/\name in {5/6} {
    \node[draw=red, circle, minimum size=0.8cm, inner sep=0pt] (n\name) at ({-360/14*\i + 90}:2cm) {\name};
  }

    \foreach \i/\name in {6/7} {
    \node[draw=blue, circle, minimum size=0.8cm, inner sep=0pt] (n\name) at ({-360/14*\i + 90}:2cm) {\name};
  }

\foreach \i/\name in {7/8} {
    \node[draw=red, circle, minimum size=0.8cm, inner sep=0pt] (n\name) at ({-360/14*\i + 90}:2cm) {\name};
  }

\foreach \i/\name in {8/9} {
    \node[draw=blue, circle, minimum size=0.8cm, inner sep=0pt] (n\name) at ({-360/14*\i + 90}:2cm) {\name};
  }

\foreach \i/\name in {9/10} {
    \node[draw=red, circle, minimum size=0.8cm, inner sep=0pt] (n\name) at ({-360/14*\i + 90}:2cm) {\name};
  }
  \node[draw=blue, circle, minimum size=0.8cm, inner sep=0pt] (n14) at ({-360/14*13 + 90}:2cm) {$14$};

  \foreach \i/\j in {
    1/2, 2/3, 3/4} {
    \draw[color=red] (n\i) -- (n\j);
  }

 \foreach \i/\j in {4/5} {
    \draw[dashed] (n\i) -- (n\j);
  }

  \foreach \i/\j in {11/12,
    12/13,13/14} {
    \draw[color=blue] (n\i) -- (n\j);
  }

\foreach \i/\j in {4/5,5/6,6/7,7/8,8/9,9/10,10/11} {
    \draw[dashed] (n\i) -- (n\j);
  }

\foreach \i/\j in {14/1} {
    \draw[dashed] (n\i) -- (n\j);
  }
\end{tikzpicture}
\caption{Graph $G_{14}$, with $A=\{1,2,3,4,6,8,10\}$ and A-open (non A-open) edges drawn with continuous (dashed) lines.}
 \label{grafo}
\end{figure}
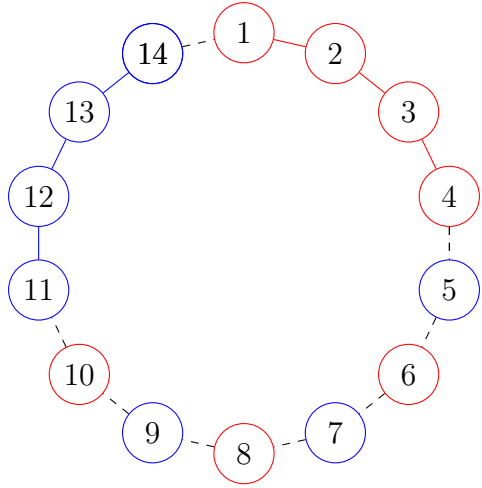

\begin{definition} \label{cluster} 
For $A \subseteq [n]$ and $i \in A$, we define the \emph{recruitment cluster} $C^-_{A}(i)$ as:
    \begin{equation}\label{def_c_meno}
          C^-_{A}(i) := 
          \begin{cases} 
             C_{A^c}(i-1) & \text{if } i > 1, \\
             C_{A^c}(n)   & \text{if } i = 1,
          \end{cases}
    \end{equation}
    where $C_{A^c}(i)$ is defined as in Equation \eqref{def_cluster}.
\end{definition}

    The cluster $C^-_{A^c}(i)$ is obtained from Equation \eqref{def_c_meno} by exchanging $A$ with $A^c$.
        Note that the sets $A$ and $\{C_A^-(k)\}_{k\in A}$ form a partition of $[n]$. \par
\begin{definition} \label{semi}
    Let $R=||r_{ji}||\in O_{m,n}$ be a voting profile. 
    We define the \emph{seed element}
    $$s_j=r_{j1}, \quad j=1,\ldots,m,$$
   and the \emph{seed vector} 
    $${\bf s}=(r_{11},r_{21},\ldots,r_{m1}).$$
\end{definition}

\begin{definition} \label{votiperi}
    Let $A \subseteq [n]$ be non-empty, and let $i \in A$. Assuming that an election is held among the elements of $A$, we define $\xi_A(i)$ as the number of votes cast for $i$ in this election.
\end{definition}
 Note that, recalling that there are no abstentions, one has
    $$
    \sum_{i\in A}\xi_A(i)=m,
    $$
and $\xi_A(i)=0$ if $i\notin A.$
\begin{lemma} \label{lemma-voti}
    Let $m,n \in \mathbb{N}$. For any   $R \in O_{m,n}$ and for any  non-empty $A \subseteq [n]$, if $i\in A$, then
    \begin{equation}
   \xi_A(i)= \sum_{j=1}^m\mathbf{I}_{C^-_{A}(i)\cup \{i\}}(s_j).
     \end{equation}
     \end{lemma}
    \begin{proof}
        Let $i\in A$. If $s_j=i$, the vote is assigned to $i$ being the preferred candidate. Otherwise, if $s_j\in C^-_{A}(i),$ by Definition \ref{cluster} we have $s_j \in A^c$. Hence, the elements of the $j$-th row of $R$ preceding $i$ are in $A^c$, therefore the first element eligible for a vote for $A$ is $i$. 
    \end{proof}

\begin{assumption}\label{assumption_1}
   For $m, n \in \mathbb{N}$, the voting profile $R \in O_{m,n}$ is randomly constructed such that each of the $m$ voters, independently from the others, chooses the seed uniformly over $[n]$.
\end{assumption}

Under Assumption \ref{assumption_1}, $\xi_A(i)$ is a random variable, denoted as $\Xi_A(i)$. Its joint distribution is multinomial:
    \begin{equation} \label{prob_unif}
        \mathsf{P}\left(\bigcap_{i \in A} \{\Xi_A(i)=k_i\}\right)=m!\cdot \prod_{i\in A}\left[\frac{1}{k_i!}\left(\frac{|C^-_{A}(i)|+1}{n}\right)^{k_i}\right],
    \end{equation}
    where $k_i\in \mathbb{N}_0$ for $i \in A$, and $\displaystyle\sum_{i\in A}k_i=m$. 

%======================
\section{Optimal Agenda Setting and Asymptotic Winning Probabilities} \label{main}
%======================  

In the following, we shall describe a voting scenario where the set of candidates is $[n]$, $n\in2\mathbb{N}$. Assume that the candidates' set is partitioned into two sets $A$ and $A^c$, with $|A|=|A^c|$. Two rounds of voting are held: in the first, all voters give a preference for two candidates, one belonging to $A$ and one to $A^c$, based on their preference list. An element $i \in A$ is declared winner of the first round if $\Xi_A(i)>\Xi_A(i')$ for any $i'\in A$ with $i'\neq i$.
Subsequently, if two winners, say $i$ and $\ell$, emerge from the first round, a second and final voting round is held between them, using the same voting profile of the first round. The candidate $i $ will be the winner of the final round if 
$$
\Xi_{\{i, \ell\}}(i) > \Xi_{\{i, \ell\}}(\ell) . 
$$
So, the probability that a candidate wins the election depends only on $n,m$ and the structure of $A$. 
Before formally defining the partition subsets, it is instructive to highlight the power of the agenda setter. From an optimization standpoint, the selection of the candidate partition can be viewed as a constrained stochastic program. Since the agenda setter must divide the candidates before observing the actual realization of the random voting profile $R \in O_{m,n}$, the strategy is strictly open-loop.  Our analytical derivation of the optimal cluster width provides the exact solution to this problem in the asymptotic regime. To intuitively grasp the implications of this mathematical framework, suppose an adversary is allowed to observe the entire realized preference matrix $R$. The adversary challenges us to ensure the victory of one specific candidate, betting on any of the others. We are entirely blind to the voters' actual preferences, but we are granted the ability to design the partition structure. Remarkably, our results demonstrate that this single degree of freedom is decisive. By optimally choosing the partition, the probability of our designated candidate winning converges to 1 as $m \to \infty$, see Corollary \ref{Cor:tutti}.

For an even number of candidates $n $ and for an integer $l$  with  $2\leq l < n/2$, we specify the two sets $A$ and $B:=A^c$  as
\begin{equation} \label{defA}
  A : = A^{(1,n,l)} =\{1,\ldots,l\}\cup C, \quad {\rm where }\;\;C:=C^{(n,l)}=  \bigcup_{i=1}^{\frac{n-2l}{2}}\{l+2i\},
\end{equation}
see Figure \ref{grafo}, where $n=14$ and $l=4$. Moreover, 
\begin{equation} \label{defB}
  B := B^{(1, n,l)} =[n]
  \setminus A^{(1,n,l)} =\{n -l +1, n -l +2 \ldots,n\}\cup C', \quad {\rm where }\;\;C':=\bigcup_{i=1}^{\frac{n-2l}{2}}\{l+2i- 1 \}. 
\end{equation}

We also define, for $i = 2, \ldots, n$,  the set
\begin{equation} \label{defAi}
  A^{(i,n,l)}: =\{[\, (a+ i-2) \mod{(n)}\,] +1 : a \in A^{(1,n,l)}  \}
\end{equation}
and $B^{(i, n,l)}: =[n] \setminus A^{(i,n,l)}$. We can think of $A^{(i,n,l)}$ as a proper rotation of the set $A^{(1,n,l)}$. \par
We will study the probability 
\begin{equation} \label{prob_f}
    p_i(n,m,l) := \mathsf{P} (F_{i,n,m,l}),
\end{equation}
where
\begin{equation*}
  F_{i,n,m,l} : =\{ \text{$i$ wins  the  two voting rounds with voting profile $R\in O_{m,n}$ and initial round $A^{(i,n,l)}$}\} .
    \label{F}
\end{equation*} 
In addition, we denote by $F^{(2)}_{n,m,l} := \bigcap_{i=1}^n F_{i,n,m,l}$ the universal victory event for all candidates, i.e. $F^{(2)}_{n,m,l} $ represents the simultaneous victory condition for all candidates, provided that the sets $A^{(i,n,l)}$ and $B^{(i,n,l)}$ are chosen according to \eqref{defA}, \eqref{defB} and \eqref{defAi}.
The winning event for candidate $i$  after the two voting rounds where the candidates are partitioned in $A^{(i,n,l)}$,  $B^{(i,n,l)}    = [n] \setminus A^{(i,n,l)}$ and there are $m $ voters.
Note that, given Assumption \ref{assumption_1}, we see that $p_i(n,m,l) $ does not depend on the index $i$. 
Then, by the union bound, we can conclude that 
\begin{equation}\label{tutti}
\mathsf{P} \Big ( F^{(2)}_{n,m,l} \Big ) \geq 1 - \sum_{i=1}^n (1- p_i(n,m,l))=1 - n (1- p_1(n,m,l)). 
\end{equation}

In order to simplify the notation of the events, in the following we will suppress the dependence on the indices $i,n,m,l$. Moreover, 
as a reference case, from now on we assume $i=1$.
For obtaining upper and lower bounds for $p_1(n,m,l)$, 
we define some relevant events. We first consider 
\begin{equation*}
  E_{1,1} :=\{ \text{Candidate 1 wins the voting in $A$}\} \cap \{ \text{Candidate $(l+1)$ wins the  voting in $B$}\} 
    \label{E11}
\end{equation*}   
\begin{equation*}
=\Big ( \bigcap_{i \in A \setminus \{1\}} \big \{ \Xi_A(i) < \Xi_A(1) \big \} \Big ) \cap \Big ( \bigcap_{i \in B \setminus \{l +1\}} \big \{ \Xi_B(i) < \Xi_B( l+1) \big \} \Big ) 
=R_{A,1}^c\cap R_{B,1}^c,
    \label{E11cap}
\end{equation*}  
where
\begin{equation}
 R_{A,1}:=\bigcup_{i \in A \setminus \{1\}} \big \{ \Xi_A(i) \geq  \Xi_A(1) \big \},
\qquad
 R_{B,1}:=\bigcup_{i \in B \setminus \{l+1\}} \big \{ \Xi_B(i) \geq  \Xi_B(l+ 1) \big \}.
 \label{eq:defRAB1}
\end{equation}
Specifically, $R_{A,1}$ occurs when candidate 1 is not winning among the candidates in $A$, and similarly $R_{B,1}$ occurs when candidate $l+1$ is not winning among the candidates in $B$. 
Moreover, we consider 
\begin{equation*}
  E_{1,2} :=\{ \text{Candidate 1 wins the voting in $A$}\} \cap \{ \text{Candidate $(l+1)$ wins the  voting in $B$}\}^c,
    \label{E12}
\end{equation*}   
\begin{equation*}
  E_{2}: =\{ \text{Candidate 1 wins  against the winner of   $B$}\} ,
    \label{E2}
\end{equation*}  
\begin{equation}
  E_{3} :=\{ \text{Candidate 1 wins  against candidate $(l+1)$}\} 
  =\big \{ \Xi_{\{  1, l+1   \}}(l+1) < \Xi_{\{  1, l+1   \}}(1) \big \}.
    \label{E3}
\end{equation} 
From now on we assume $i=1$ and define $F := F_{1, n,m,l}$. 
The event $F$
can be expressed as 
\begin{equation} \label{def_F}
    F = (E_{1,1}\cup E_{1,2}) \cap E_2 = (E_{1,1 } \cap E_3) \cup (E_{1,2}\cap E_2), 
\end{equation}
since $E_{1,1 } \cap E_2=E_{1,1 } \cap E_3$.

In the next lemma we focus on the event $F^c$.

\begin{lemma} \label{subset}
Recalling \eqref{eq:defRAB1} and \eqref{def_F}, for any voting profile in $O_{m,n}$ the following inclusions hold:
\begin{align}
     F^c &\subset  R_{A,1} \cup R_{B,1} \cup E_{3}^c; 
     \label{sub1} 
     \\
        F^c& \supset   \big \{ \Xi_A(i) \geq  \Xi_A(1) \big \},   \qquad \forall\; i \in A\setminus\{1\};%\nonumber 
    \label{sub2}    \\
        F^c& \supset  E_3^c \setminus 
       R_{B,1}.  %\nonumber
 \label{sub3}   
 \end{align}
\end{lemma}
\begin{proof} We  start by proving \eqref{sub1}. 
Recalling (\ref{eq:defRAB1}), \eqref{E3} and  (\ref{def_F}), we have 
\begin{eqnarray*}
  F^c & = & \Big  (( E_{1,1}\cap E_3 ) \cup (E_{1,2} \cap E_2) \Big )^c 
  = ( E_{1,1}\cap E_3 )^c \cap  (E_{1,2} \cap E_2)^c 
  \subset ( E_{1,1}\cap E_3 )^c 
\\
 &= & E_{1,1}^c \cup E_3^c
 = R_{A,1} \cup R_{B,1} \cup E_3^c.
\end{eqnarray*}
In order to prove the inclusion in  \eqref{sub2} it is enough to observe that if $\Xi_A(i) \geq  \Xi_A(1)$ then 
  the candidate $1 $ does not win the first round. In conclusion, let us now prove \eqref{sub3}. We have $ E_3^c \setminus 
       R_{B,1}= E_3^c \cap  R_{B,1}^c$. Therefore, for any $\omega \in E_3^c \cap 
       R_{B,1}^c$, candidate $(l+1) $ is the winner of the first round in $B$. 
 Since $ \omega \in  E_3^c $, candidate $1$ does not win over candidate $(l+1)$ when voting among the set $\{1,l+1\}$. Hence, candidate 1 cannot be the winner of the two rounds, thus $\omega \in F^c$. This ends the proof. 
\end{proof}

\medskip 

For the reader's convenience, we start by presenting a particular application of the Chernoff inequality and Cramér's theorem (cf.\ \cite{dembo}) to our framework.
\begin{lemma}\label{chernoff_bound}
    Let $m \in \mathbb{N}$ and $Y^{(m)}$ be a random variable defined as
    \begin{equation} \label{def_y}
        Y^{(m)}:=\sum_{j=1}^m Y_j,
    \end{equation}
    where $\{Y_1,Y_2,\ldots,Y_m\}$ represents a family of  i.i.d. random variables such that
    \begin{align}
        p_{-1}:=\mathsf{P}(Y_j=-1)&; \qquad
        p_{0}:=\mathsf{P}(Y_j=0); \qquad 
        p_{1}:=\mathsf{P}(Y_j=1); \nonumber\\ 
        p_{-1}&<p_{1}, \quad p_{-1}+p_{0}+p_{1}=1. \label{cond_p}
    \end{align}
    Then,
    \begin{equation}\label{upper_bound}
        {\mathsf{P}}(Y^{(m)}\leq 0)\leq \exp\{ m \log {(1-(\sqrt{p_{1}}-\sqrt{p_{-1}})^2)}\}.
    \end{equation}
Furthermore,
    \begin{equation} \label{lim_log_1}
\lim_{m \to \infty }\frac{ \log ( {\mathsf{P}}(Y^{(m)}\leq 0))}{m} =\log {(1-(\sqrt{p_{1}}-\sqrt{p_{-1}})^2)}.
\end{equation}
\end{lemma}
\begin{proof}
We recall from large deviations theory that the rate function, i.e. the Legendre transform of the moment generating function of $Y_1$, is
    \begin{equation*}
        I(x) = \sup_{t \in \mathbb{R}} \{ t x -\log \mathsf{E}\left[\exp\{t\ Y_1\}\right] \}.
    \end{equation*}
    In order to obtain the expressions in \eqref{upper_bound} and \eqref{lim_log_1}, we take $x=0$ in the previous definition. Thus, we obtain
\begin{equation}\label{ratefunct}
        I(0) = \sup_{t\in \mathbb{R}} \{-\log \mathsf{E}\left[\exp\{t\ Y_1\}\right]\}=\sup_{t<0} \{- \log\mathsf{E}\left[\exp\{t\ Y_1\}\right]\},
    \end{equation}
    where $t$ is restricted to negative values since
    $\mathsf{E}[Y_1]=p_1-p_{-1}>0$ (see \eqref{cond_p}). 
    Since
    \begin{align*}
        \frac{{\rm d}}{{\rm d} t}\left\{-\log\mathsf{E}\left[\exp\{t\ Y_1\}\right]\right\}
        % =\frac{e^{-t} p_{-1}-e^{t}p_{1}}{e^{t}p_{1}+e^{-t}p_{-1}+p_0}
        =0 \quad \Leftrightarrow \quad t=t_{max}:=\log \left(\sqrt{\frac{p_{-1}}{p_1}}\right)<0,
    \end{align*}
    and
    $\mathsf{E}\left[\exp\{t\ Y_1\}\right]=e^{t}p_{1}+e^{-t}p_{-1}+p_0,$
   we get $$I(0)=-\log\mathsf{E}[t_{max}\ Y_1]=-\log {(1-(\sqrt{p_{1}}-\sqrt{p_{-1}})^2)}.$$
    From Chernoff bound we get 
    \begin{equation} \label{bound_Y}
        \mathsf{P}(Y^{(m)}\leq 0)\leq \exp\Big\{-m\ I(0)\Big\},
    \end{equation}
    which leads to \eqref{upper_bound}.
Furthermore, from Cramér's theorem we get \eqref{lim_log_1}. 
\end{proof}
In the following lemma, given two random variables $Y^{(m)}$ and $\widetilde{Y}^{(m)}$ defined as in \eqref{def_y}, we provide a sufficient condition for the usual stochastic order.
\begin{lemma}\label{Confronto-chernoff}
    For $m\in\mathbb{N}$, let $Y^{(m)}$ and $\widetilde{Y}^{(m)}$ be two random variables defined as 
    $$Y^{(m)}:=\sum_{j=1}^m Y_j, \qquad \widetilde{Y}^{(m)}:=\sum_{j=1}^m \widetilde{Y}_j,$$
    where $(Y_j)_{j\in \mathbb{N}}$ and $(\widetilde{Y}_j)_{j\in \mathbb{N}}$ are two families of i.i.d. random variables. Suppose that $Y_1$ is stochastically larger than $\widetilde{Y}_1$, i.e.
    \begin{align}
        p_{-1}&:=\mathsf{P}(Y_1=-1)\leq \mathsf{P}(\widetilde{Y}_1=-1)=:  \widetilde{p}_{-1} ;\qquad 
        p_{0}:=\mathsf{P}(Y_1=0); 
        \nonumber\\     p_{1}&:=\mathsf{P}(Y_1=1)\geq\mathsf{P}(\widetilde{Y}_1=1)=:  \widetilde{p}_{1};\qquad \widetilde{p}_{0}:=\mathsf{P}(\widetilde{Y}_1=0); \nonumber \\
        p_{-1}&+p_{0}+p_{1}= \widetilde{p}_{-1}+\widetilde{p}_{0}+\widetilde{p}_{1}=1. \nonumber
    \end{align}
    Then, $Y^{(m)}$ is larger than $\widetilde{Y}^{(m)}$ in the usual stochastic order, i.e.
    $${\mathsf{P}}(Y^{(m)}\leq y)\leq {\mathsf{P}}(\widetilde{Y}^{(m)}\leq y), \qquad \forall\; y \in \mathbb{R}.$$
\end{lemma}
\begin{proof}
   The result is a particular case of Theorem 1.A.3 of \cite{ShShbook}.
\end{proof}
The following lemma establishes an upper bound for the probability that candidate $1$ does not win over candidate $i$, with $1,i \in A$.
\begin{lemma}\label{1-su-i}
  Let $n\geq 6$, $n\in 2\mathbb{N}$, be the number of candidates, $l $ the integer used in \eqref{defA} to 
  define $A $ and  $m$ the number of voters. Then, under Assumption \ref{assumption_1}, for any $i \in A\setminus\{1\}$ one has
  $$
 \mathsf{P} \big (\Xi_A (1) \leq \Xi_A (i) \big ) \leq \exp \left \{ m\ \log \left [1-  \frac{1}{n}\left  ( \sqrt{ {l+1}}-\sqrt{2}\right )^2\right ]   \right \}.
  $$ 
\end{lemma}
\begin{proof} 
Let us start  by observing that $C^-_{A}(1)=\{n -l +1, n -l +2 ,\ldots,n\}$, due to Definition \ref{cluster} and Equation (\ref{defB}). Therefore, by Lemma \ref{lemma-voti}, the probability that a generic voter assigns the vote to candidate 1 is 
\begin{equation}
\frac{|C^-_{A}(1)|+1}{n} = \frac{l+1}{n}, 
\label{eq:C1m}
\end{equation}
regardless of the behavior of all other voters. 
For any $i \in A \setminus\{1\}$ there are two possible cases determined by Lemma \ref{lemma-voti}, where 
assumption $n\geq 6$ ensures that $C$ is not empty, due to \eqref{defA}.
\begin{itemize}
    \item[(i)] If $i \in \{ 2, 3, \ldots , l \}$ then $C^-_{A}(1) = \emptyset $ and the probability that a generic voter assigns the vote to candidate $i$ is $1/n$; 
     \item[(ii)]  if $i \in C $ then $C^-_{A}(1) = \{ i-1\} $ and  the probability that a generic voter assigns the vote to candidate $i$ is $2/n$. 
\end{itemize}
     Hence, for any $i \in A \setminus \{1\}$ and $j \in [m] $ we define the random variable 
     \begin{equation}\label{Xj}
         X_{i,j} :=  
        \left\{ \begin{array}{rll}
1, &\text{ if } v_j \text{ votes for candidate } 1; \\ 
-1, &\text{ if }  v_j \text{ votes for candidate } i;\\
0, &\text{ otherwise.}  \\
\end{array}
\right .
\end{equation}
From Lemma \ref{lemma-voti}, Equation (\ref{eq:C1m}) and the above items (i) and (ii) we have 
\begin{equation}\label{probXij}
     \mathsf{P} ( X_{i,j} =1) = \frac{l+1}{n},
     \qquad 
     \mathsf{P} ( X_{i,j} =-1) = \left\{
     \begin{array}{ll}
     \displaystyle\frac{1}{n}, & i\in\{2,\ldots, l\}\\[3mm]
     \displaystyle\frac{2}{n}, & i\in C.
     \end{array}
     \right.
\end{equation}
By Assumption \ref{assumption_1}, for any fixed $i\in A\setminus\{1\}$, the random variables $(X_{i,j})_{j \in [m]}$ are independent and identically distributed. 
Furthermore, one has 
\begin{align}
\Xi_A (1) -\Xi_A (i) = \sum_{j =1}^m X_{i,j}.
\label{eq:XiXij}
\end{align}
From Lemma \ref{chernoff_bound},  one obtains 
\begin{equation*}
 \mathsf{P} \left ( \sum_{j =1}^m X_{i,j} \leq 0 \right ) \leq 
 \exp \left \{ m\ \log \left [1- \frac{1}{n} 
 \left ( \sqrt{l+1}-1 \right)^2\right ]   \right \}, \quad i \in \{2, \ldots , l\},
\end{equation*} 
and 
\begin{equation*}
  \mathsf{P} \left ( \sum_{j =1}^m X_{i,j} \leq 0 \right ) \leq 
 \exp \left \{ m\ \log \left [1- \frac{1}{n}
 \left ( \sqrt{l+1}-\sqrt{2}\right )^2\right ]   \right \}, \quad i \in C.
\end{equation*} 
Therefore, for any $i \in A\setminus\{ 1 \}$ we have 
$$
\mathsf{P} \big (\Xi_A (1) -\Xi_A (i)  \leq 0 \big ) 
=  \mathsf{P} \left ( \sum_{j =1}^m X_{i,j} \leq 0 \right ) \leq 
 \exp \left \{ m\ \log \left[1-\frac{1}{n}\left(\sqrt{l+1}-\sqrt{2}\right )^2\right ]   \right \}, 
$$
since
$$
\exp \left \{ m\ \log \left [1-\frac{1}{n}\left  ( \sqrt{{l+1}}-1\right )^2\right ]   \right \}
\leq \exp \left \{ m\ \log \left [1-\frac{1}{n}\left  ( \sqrt{l+1}-\sqrt{2}\right )^2\right ]   \right \}.
$$
This completes the proof. 
\end{proof}

Now, we present a similar result for the candidate $(l+1) \in B$. The technical details are omitted because of the similarity to the steps presented in Lemma \ref{1-su-i}.
\begin{lemma}\label{l+1-su-i}
Let $n\geq 6$, $n \in 2\mathbb{N}$, represent the number of candidates,  $l $ the integer used in \eqref{defB} to define $B$ and $m$ the number of voters. Then, under Assumption \ref{assumption_1}, for any $i \in B\setminus\{l+1\}$, we have
$$
 \mathsf{P} \big (\Xi_B (l+1) \leq \Xi_B (i) \big ) \leq \exp \left \{ m\ \log \left [1-\frac{1}{n}\left  ( \sqrt{l+1}-\sqrt{2}\right )^2\right ]   \right \}.
$$ 
\end{lemma}
\begin{proof} 
The proof follows analogously to Lemma \ref{1-su-i} by making use of the sets $C^-_{B}(i)$, with $i\in B$, and noting from Lemma \ref{lemma-voti} that $|C^-_{B}(l+1)|=|C^-_{A}(1)|=l.$ 
\end{proof}
\begin{lemma}\label{due_candidati}
Let $n\geq 6$, $n \in 2\mathbb{N}$, represent the number of candidates and let $m$ be the number of voters. 
Then, under Assumption \ref{assumption_1}, for any $k  \in \{2, \ldots , n/2 \}$, we have
$$
 \mathsf{P} \big (\Xi_{\{1, k\}} (1) \leq \Xi_{\{1, k\}} (k) \big ) \leq \exp \left \{ m\ \log \left [1-\frac{1}{n}\left  ( \sqrt{n+1-k}-\sqrt{k-1}\right )^2\right ]   \right \}.
$$ 
\end{lemma}
\begin{proof} 
The proof follows analogously to Lemma \ref{1-su-i} by making use of the sets $C^-_{\{1,k\}}(i)$, with $i\in B$, and noting from Lemma \ref{lemma-voti} that $|C^-_{\{1,k\}}(k)|=k-2.$ 
\end{proof}
We are now able to study the behavior of the winning probability 
of candidate 1 when the number of voters $m$ grows.
\begin{theorem} \label{limsup_1}
 Under Assumption \ref{assumption_1}, if $n\geq 6$ is an even integer and $2\leq l\leq \frac{n}{2}-1,$ then 
\begin{equation} 
\limsup_{m\to \infty}\frac{\log(1-p_1(n,m,l))}{m}=:c<0.
 \label{eq:defc}
\end{equation} 
\end{theorem}
\begin{proof}
     By Eq.\ \eqref{sub1} of Lemma \ref{subset}, from Lemmas \ref{1-su-i}--\ref{due_candidati}, making use of \eqref{eq:defRAB1} and \eqref{E3}, one has 
\begin{align}
    &1-p_1\left(n,m,l\right)\leq \mathsf{P} \left (  R_{A,1} \cup R_{B,1} \cup E_3^c \right)\leq \mathsf{P} (R_{A,1})+ \mathsf{P} (R_{B,1}) +\mathsf{P} (E_3^c) \nonumber\\ 
    &\leq \sum_{i \in A \setminus \{1\}} \mathsf{P} (\Xi_A(i) \geq  \Xi_A(1)   ) +\sum_{i \in B \setminus \{l+1\}} \mathsf{P} (\Xi_B(i) \geq  \Xi_B(l+1)   ) +\mathsf{P} \left ( \Xi_{\{  1, l+1   \}}(l+1) \geq \Xi_{\{  1, l+1   \}}(1)  \right )  \nonumber \\ 
    &\leq (n-2)\exp \left \{ m\ \log \left [1-\frac{1}{n}\big( \sqrt{l+1}-\sqrt{2}\big)^2\right ]   \right \}+\exp\left\{ m\ \log \left[1-\frac{1}{n}\big( \sqrt{n-l}-\sqrt{l}\big)^2\right] \right\}  \nonumber \\
    &\leq (n-1)\exp\left\{ m\ \log\left[ \max\left\{ 1-\frac{1}{n}\big( \sqrt{l+1}-\sqrt{2}\big)^2,1-\frac{1}{n}\big( \sqrt{n-l}-\sqrt{l}\big)^2\right\}\right]\right\} \nonumber \\
    &=\exp\left\{m\left[ \frac{\log(n-1)}{m}+\log\left( \max \left\{ 1-\frac{1}{n}\big( \sqrt{l+1}-\sqrt{2}\big)^2,1-\frac{1}{n}\big( \sqrt{n-l}-\sqrt{l}\big)^2\right\}\right)\right]\right\} .  \label{diseq_1}
\end{align}
Since $\lim_{m\to \infty} \frac{\log(n-1)}{m}=0$, we have 
\begin{align} \label{lim_log}
    \limsup_{m\to \infty}\frac{\log(1-p_1(n,m,l))}{m}\leq
    %\limsup_{m\to +\infty} \;
    \log\left( \max \left\{ 1-\frac{1}{n}\big( \sqrt{l+1}-\sqrt{2}\big)^2,1-\frac{1}{n}\big( \sqrt{n-l}-\sqrt{l}\big)^2\right\}\right).
\end{align}
Since $2\leq l \leq \frac{n}{2}-1$, we see that the logarithm on the right-hand side of \eqref{lim_log} is negative, proving the result.
\end{proof}
From Theorem \ref{limsup_1}, we obtain that the winning probability for candidate 1 after the two voting rounds tends to 1 exponentially fast in the number of voters $m$. 
\begin{corollary} \label{Cor:tutti}
Let $c'\in (c,0)$, with $c$ defined in (\ref{eq:defc}).  
Under the assumptions of Theorem \ref{limsup_1}, for any $m\in \mathbb N$ and some constant $K>1$ non dependent on $n$, 
we have   
\begin{equation}\label{esponenziale}
 p_1(n,m,l)\geq 1- K\, e^{c' m}.
\end{equation}
Furthermore, 
\begin{equation}\label{esponenziale2}
\mathsf{P} \Big (\bigcap_{i=1}^n F_{i,n,m,l} \Big ) \geq 1 - n K\, e^{c' m}. 
\end{equation}
\end{corollary}
\begin{proof}
Eq.\ \eqref{esponenziale} is equivalent to    
$$
\frac{\log(K) }{m}+c'  \geq \frac{\log(1- p_1(n,m,l))}{m}.
$$
From Theorem \ref{limsup_1} we know that there exists an integer $\overline{m} = \overline{m} (c')$ such that, for all $ m\geq \overline{m}$, 
$$
c'  \geq \frac{\log(1- p_1(n,m,l))}{m}.
$$
Now, taking $K = e^{-c' \overline{m} } >1$ we obtain that for any $m \in \N $ the inequality in \eqref{esponenziale} holds. 
Finally, the relation in \eqref{esponenziale2} follows from Eqs.\ \eqref{tutti} and \eqref{esponenziale}.
\end{proof}
\par
For given integers $n$ and $m$, we define the \emph{optimal values of the width} $l$ as
\begin{equation} \label{L_opt}
     L_{opt}(n,m):={\rm arg\, max}_{l \in \{ 2,\ldots,  \frac{n}{2}-1 \} }\, p_1(n,m,l).
\end{equation}
Clearly, the set $L_{opt}(n,m)$ can contain multiple values. Hence, for any $l_{n,m}$ in $L_{opt}(n,m)$, we shall also investigate the \textit{relative optimal width} $\displaystyle \frac{l_{n,m}}{n}$.

The following theorem, which is the main result of the paper, yields an asymptotical expression concerning the relative optimal width. Specifically, we assume that the number of voters $m$ is a function of the number of candidates $n$ such that its growth to infinity is much faster than $\log n$. Under this assumption, we find that the relative optimal width has a specific constant limit, this being useful to maximize the winning probability for each candidate. 

\begin{theorem} \label{Bound}
    Let $M:\mathbb{N}\to \mathbb{N}$ be a function such that 
\begin{equation}
   \lim_{n \to \infty} \frac{M(n)}{\log n}=\infty.
    \label{eq:limMlog}
\end{equation}  
Suppose that Assumption \ref{assumption_1} holds, and consider a sequence $(l_{n,M(n)})_{n \in \N}$ eventually satisfying $l_{n,M(n)} \in L_{opt}(n,M(n))$.
    Then, one has
    \begin{equation} \label{lim_l} 
            \lim_{n\to \infty} \frac{l_{n,M(n)}}{n}=\frac{1}{5}.
    \end{equation}
Moreover, the following upper bound for the exponential decay holds:
    \begin{equation} \label{limsup_prob}
            \limsup_{n \to \infty} \frac{1}{M(n)}\log\left[1-p_1\left(n,M(n),l_{n,M(n)}\right)\right] \leq \log\left(\frac{4}{5}\right).
    \end{equation}
\end{theorem}
\begin{proof}
The proof proceeds by contradiction. 
From Equation \eqref{L_opt} one gets  
$$
 L_{opt}(n,m):=\frac{1}{m}{\rm arg\, min}_{l \in \{2,\ldots, \frac{n}{2}-1\}}\, \log \left[1-p_1(n,m,l)\right].
$$
Now we begin by finding an upper bound of
$$
\frac{1}{M(n)}\log\left[1-p_1\left(n,M(n),l_{n,M(n)}\right)\right]
$$
by means of the Chernoff bound. For the reader's convenience, we recall Equation \eqref{diseq_1}, i.e.\ 
\begin{align}
    &1-p_1\left(n,m,l\right) \nonumber \\
    &\leq \exp\left\{m\left[ \frac{\log(n-1)}{m}+\log\left( \max \left\{ 1-\frac{1}{n}\big( \sqrt{l+1}-\sqrt{2}\big)^2,1-\frac{1}{n}\big( \sqrt{n-l}-\sqrt{l}\big)^2\right\}\right)\right]\right\} .  \nonumber
\end{align}
Hence, by choosing $m=M(n)$, where $M(n)$ satisfies assumption (\ref{eq:limMlog}), and $ l= \widehat{l}_n $ where  $(\widehat{l}_n)_{n \in \N}$ is a sequence such that 
$ \lim_{n\to \infty}\frac{\widehat{l}_n}{n}=\frac{1}{5} $, 
one has 
\begin{align} 
&\limsup_{n\to \infty} \frac{1}{M(n)}\log\left[1-p_1\left(n,M(n),  \widehat{l}_n\right)  \right]
\nonumber\\
& \leq \limsup_{n\to \infty} 
\left[ \frac{\log(n-1)}{M(n)}
+\log\left( \max \left\{ 1-\frac{1}{n}
\left( \sqrt{\widehat{l}_n+1}-\sqrt{2}\right)^2,
1-\frac{1}{n}\left( \sqrt{n-\widehat{l}_n}-\sqrt{\widehat{l}_n}\right)^2\right\}\right)\right]
\nonumber\\
& =
\log\left(\limsup_{n\to \infty}  \max \left\{ 1-\frac{1}{n}
\left( \sqrt{\frac{n}{5}+o(n)}-\sqrt{2}\right)^2,
1-\frac{1}{n}\left( \sqrt{n-\frac{n}{5}+o(n)}
-\sqrt{\frac{n}{5}+o(n)}\right)^2\right\}\right)
\nonumber\\
& =
\log\left(1-\min \left\{\lim_{n\to \infty}   \frac{1}{n}
\left( \sqrt{\frac{n}{5}+o(n)}-\sqrt{2}\right)^2,
\lim_{n\to \infty} \frac{1}{n}\left( \sqrt{ \frac{4}{5}\,n+o(n)}
-\sqrt{\frac{n}{5}+o(n)}\right)^2\right\}\right)
\nonumber\\
& = \log\left(\frac{4}{5}\right).
\label{limsup_log}
\end{align}
If the thesis of the theorem were false, then there would exist a sequence  $ (\widetilde{l}_n)_{n \in \N}$    with    $ \widetilde{l}_n \in L_{opt}(n,M(n))$ such that
$$
(i) \quad \liminf_{n \to \infty} \frac{\widetilde{l}_n}{n}< \frac{1}{5}
\qquad \vee \qquad
(ii)\quad \limsup_{n \to \infty} \frac{\widetilde{l}_n}{n}> \frac{1}{5}.
$$
Let us analyze case $(i)$.  Then, up to a subsequence, one
has
\begin{equation} \label{l_k}
  \lim_{k \to \infty} \frac{\widetilde{l}_{n_k}}{n_k}    = \liminf_{n \to \infty} \frac{\widetilde{l}_n}{n}= :p < \frac{1}{5} \,.   
\end{equation}
 However, we will show that any such sequence $(\widetilde{l}_{n_k})$ satisfies
 $$
 \liminf_{k\to \infty} \frac{1}{M(n_k)}\log\left[1-p_1\left(n_k,M(n_k),\widetilde{l}_{n_k}\right)  \right]> \log\left(\frac{4}{5}\right) \,.
 $$
This inequality, together with \eqref{limsup_log}, implies that there exist infinite values of $n \in \mathbb{N}$ such that $\widetilde{l}_n \notin L_{opt}(n,M(n))$. This is in contradiction with the hypothesis of the sequence $(l_{n,M(n)})_{n \in \N}$ eventually satisfying $l_{n,M(n)} \in L_{opt}(n,M(n))$.

\medskip 

$\bullet$ By Equation \eqref{sub2} of Lemma \ref{subset}, with $i=2\in A$, one has 
$$
1- p_1(n_k,M(n_k), \widetilde l_{n_k} ) \geq \mathsf{P} \left( \big \{ \Xi_A(2) \geq  \Xi_A(1) \big \}\right).
$$
Furthermore, 
\begin{align*}
\liminf_{k \to \infty } \frac{1}{M(n_k)}\log[ 1- p_1(n_k,M(n_k),  \widetilde l_{n_k} )] 
&\geq 
\liminf_{k \to \infty } \frac{1}{M(n_k)}\log[ \mathsf{P} ( \big \{ \Xi_A(2) \geq  \Xi_A(1) \big \}  ) ]
\\
&= 
\liminf_{k \to \infty }  \frac{1}{M(n_k)}
\log\Bigg[ \mathsf{P}\Big( \sum_{j=1}^{M(n_k)} X_{2,j} \leq 0 \Big)\Bigg],
\end{align*}
where random variables $(X_{2,j} : j \in  [M(n_k)])$ are defined in the proof of Lemma \ref{1-su-i}, cf.\ Eq.\ (\ref{eq:XiXij}).

By  Equations \eqref{Xj} and \eqref{probXij}, we recall that, for any $j \in [M(n_k)] $, the random variable $ X_{2,j}$
satisfies
\begin{equation*}
     \mathsf{P} ( X_{2,j} =1) = \frac{\widetilde{l}_{n_k}+1}{n_k},
     \qquad 
     \mathsf{P} ( X_{2,j} =-1) = \frac{1}{n_k}.
\end{equation*}
Next, consider another sequence of i.i.d. random variables $(X'_{j})_{j \in \mathbb{N}}$ satisfying
$$
\mathsf{P} ( X'_{j} =1)=p'_{1}   \in \left(p, \frac{1}{5}\right), 
\qquad \mathsf{P} ( X'_{j} =0)=1-p'_{1}. 
$$
Then, there exists $k_0 \in \mathbb{N}$ such that for all $k \geq k_0$, it holds that
$$
\frac{\widetilde{l}_{n_k}+1}{n_k} < p'_1 .
$$
Notice that $X_{2,j} $ is smaller than $X'_j$ in the usual stochastic order. Then, applying Lemma \ref{Confronto-chernoff}, we obtain
\begin{align*} 
\liminf_{k \to \infty }  \frac{1}{M(n_k)}\log\Bigg[ \mathsf{P}\Bigg( \sum_{j=1}^{M(n_k)} X_{2,j} \leq 0 \Bigg)\Bigg] 
&\geq 
\liminf_{k \to \infty }  \frac{1}{M(n_k)}\log\Bigg[ \mathsf{P}\Bigg( \sum_{j=1}^{M(n_k)} X'_{j} \leq 0 \Bigg)\Bigg] 
\nonumber \\
& = \log \left (1- p'_1 \right) > 
\log \left ( \frac{4}{5} \right ), 
\end{align*}
where the equality follows from large deviation theory (see \eqref{lim_log_1} of Lemma \ref{1-su-i}), while the 
last inequality follows from $p'_{1}   \in \left(p, \frac{1}{5}\right)$.

\medskip 

Let us analyze case $(ii)\quad  \widetilde p : = \limsup_{n \to \infty} \frac{\widetilde{l}_n}{n}> \frac{1}{5}$. 
Then, up to a subsequence, one has  
\begin{equation} \label{r_k}
  \lim_{k \to \infty} \frac{\widetilde{l}_{m_k}}{m_k} = \widetilde p > \frac{1}{5}.  
\end{equation}
Again, we will show that under condition \eqref{r_k} one has 
$$
\liminf_{k\to \infty} \frac{1}{M(m_k)}\log\left[1-p_1\left(m_k,M(m_k),\widetilde{l}_{m_k}\right)\right]
> \log\left(\frac{4}{5}\right). 
$$
By Equation \eqref{sub3} of Lemma \ref{subset}, one has 
\begin{align}\label{diseq_p1}
    1- p_1(m_k,M(m_k), \widetilde{l}_{m_k} ) &\geq \mathsf{P} \big(  \Xi_{\{1, \widetilde{l}_{m_k}+1\}}(\widetilde{l}_{m_k}+1) \geq  \Xi_{{\{1, \widetilde{l}_{ m_k}+1}  \}} (1) \big   )\nonumber\\
    &- \mathsf{P} \Bigg(  \bigcup_{i \in B \setminus 
\{\widetilde{l}_{m_k} +1\}} \big \{ \Xi_B(i) \geq \Xi_B( \widetilde{l}_{m_k}+1) \big \} \Bigg). 
\end{align}
For the union bound, we get that
\begin{align*}
    \mathsf{P} \Bigg(  &\bigcup_{i \in B \setminus 
\{\widetilde{l}_{m_k} +1\}} \big \{ \Xi_B(i) \geq \Xi_B( \widetilde{l}_{m_k}+1) \big \} \Bigg)\leq  \sum_{i \in B \setminus 
\{\widetilde{l}_{m_k} +1\}} \mathsf{P}\left( \Xi_B(i) \geq \Xi_B( \widetilde{l}_{m_k}+1) \right) \\ 
&\leq        \sum_{i \in B \setminus \{\widetilde{l}_{m_k} +1\}}  \exp\left\{ M(m_k) \cdot \log\left[  1-\frac{1}{m_k}\big( \sqrt{\widetilde{l}_{m_k}+1}- \sqrt{2}\Big)^2
\right] \right\}  
\end{align*}
by Lemma \ref{l+1-su-i}. 
Hence, the left-hand side of Eq.\ \eqref{diseq_1} is larger or equal than
\begin{align*}
    \mathsf{P} \Big (  \Xi_{\{1, \widetilde{l}_{m_k}+1\}}(\widetilde{l}_{m_k}+1  ) \geq  \Xi_{\{1, \widetilde{l}_{m_k}+1\}} (1) \Big  ) - \frac{m_k}{2} \exp{ \Big\{M(m_k) \cdot \log\left[  1-\frac{1}{m_k}\big( \sqrt{\widetilde{l}_{m_k}+1}- \sqrt{2}\Big)^2
\right] \Big \}}  \\ 
\geq   \mathsf{P} \big ( \Xi_{\{1, \widetilde{l}_{m_k}+1\}}(\widetilde{l}_{m_k}+1) \geq  \Xi_{\{1, \widetilde{l}_{m_k}+1\}} (1) \big  ) - \Big( \frac{4}{5}\Big)^{M(m_k)}  ,
\end{align*}
where the last inequality follows from Lemma \ref{l+1-su-i} and inequality \eqref{r_k}.

Furthermore, 
\begin{align} \label{prob_diseq}
\liminf_{k \to \infty } &\frac{1}{M(m_k)}\log[ 1- p_1(m_k,M(m_k), \widetilde{l}_{m_k} )] \nonumber  \\
&\geq 
\liminf_{k \to \infty } \frac{1}{M(m_k)}\log \Big [ \mathsf{P} \big ( \Xi_{\{1, \widetilde{l}_{m_k}+1\}}(\widetilde{l}_{m_k}+1) \geq  \Xi_{\{1, \widetilde{l}_{m_k}+1\}} (1) \big   ) - \Big( \frac{4}{5}\Big)^{M(m_k)}       \Big ] 
\nonumber \\
& =
\liminf_{k \to \infty }  \frac{1}{M(m_k)}\log\Bigg[ \mathsf{P}\Bigg( \sum_{j=1}^{M(m_k)} X_{\widetilde{l}_{m_k}+1,j} \leq 0 \Bigg)    
- \Big( \frac{4}{5}\Big)^{M(m_k)}  
\Bigg] ,
\end{align}
where the random variables $(X_{\widetilde{l}_{m_k}+1,j} : j \in  [M(m_k)] ) $ are defined in the proof of Lemma \ref{chernoff_bound}. 
\par
By making use of Equations \eqref{Xj} and \eqref{probXij}, we recall that, for any $j \in [M(m_k)] $, the random variable 
     \begin{equation*}
         X_{\widetilde{l}_{m_k} +1,j} =  
        \left\{ \begin{array}{rll}
1, &\text{ if } v_j \text{ votes for candidate } 1; \\ 
-1, &\text{ if }  v_j \text{ votes for candidate } \widetilde{l}_{m_k} +1;\\
0, &\text{ otherwise}  \\
\end{array}
\right .
     \end{equation*}
is such that
\begin{equation*}
     \mathsf{P} \left( X_{\widetilde{l}_{m_k} +1,j} =1\right) = \frac{m_k -\widetilde{l}_{m_k}}{m_k},
     \qquad 
     \mathsf{P} \left( X_{\widetilde{l}_{m_k} +1,j} =-1\right) = 
     \displaystyle\frac{\widetilde{l}_{m_k}}{m_k}.
\end{equation*}
Let us now consider a sequence of i.i.d. random variables $(\widetilde{X}_{j})_{j \in \mathbb{N}}$ satisfying
$$
\mathsf{P} ( \widetilde{X}_{j} =-1)=\widetilde{p}_{-1}   \in \left(\frac{1}{5}, \min\left\{\frac{4}{5}, \widetilde{p}\right\}\right), 
\qquad \mathsf{P} ( \widetilde{X}_{j} =1)=1-\widetilde{p}_{-1}. 
$$
Notice that $X_{\widetilde{l}_{m_k}+1,j} $ is smaller than $\widetilde{X}_j$ in the usual stochastic order for a sufficiently large $k$. Hence, recalling Eq.\ \eqref{prob_diseq} and by means of Lemma \ref{Confronto-chernoff}, we get that
\begin{align} \label{liminf_p}
    &\liminf_{k \to \infty }  \frac{1}{M(m_k)}\log\Bigg[ \mathsf{P}\Bigg( \sum_{j=1}^{M(m_k)} X_{\widetilde{l}_{m_k}+1,j} \leq 0 \Bigg)    
- \Big( \frac{4}{5}\Big)^{M(m_k)}  
\Bigg] \geq \nonumber\\ 
&\liminf_{k \to \infty }  \frac{1}{M(m_k)}\log\Bigg[ \mathsf{P}\Bigg( \sum_{j=1}^{M(m_k)} \widetilde{X}_{j} \leq 0 \Bigg)    
- \Big( \frac{4}{5}\Big)^{M(m_k)}  
\Bigg]\geq \nonumber\\
&\liminf_{r \to \infty }  \frac{1}{r}\log\Bigg[ \mathsf{P}\Bigg( \sum_{j=1}^{r} \widetilde{X}_{j} \leq 0 \Bigg)    
- \Big( \frac{4}{5}\Big)^{r}  
\Bigg].
\end{align}
From Eq.\ \eqref{lim_log_1}, one has that
\begin{equation*} 
{\mathsf{P}}\left(\sum_{j=1}^{r} \widetilde{X}_{j} \leq 0 \right) =\left[1-\left(\sqrt{1-\widetilde{p}_{-1}}-\sqrt{\widetilde{p}_{-1}}\right)^2\right]^{r+ o(r)}.
\end{equation*}
Then, we notice that
\begin{equation} \label{diseq_p-1}
    1-\left(\sqrt{1-\widetilde{p}_{-1}}-\sqrt{\widetilde{p}_{-1}}\right)^2>\frac{4}{5},
\end{equation}
or equivalently
$$ \widetilde{p}_{-1}^2-\widetilde{p}_{-1}+\frac{4}{25}<0,
$$
holds true for any $\widetilde{p}_{-1}   \in \left(\frac{1}{5}, \min\left\{\frac{4}{5}, \widetilde{p}\right\}\right)$. Thus, from Eq.\ \eqref{diseq_p-1}, the last expression of \eqref{liminf_p} becomes
\begin{align*}
    \liminf_{r \to \infty }  \frac{1}{r}\log\left[ \mathsf{P}\Bigg( \sum_{j=1}^{r} \widetilde{X}_{j} \leq 0 \Bigg)\right]& =\liminf_{r \to \infty }  \frac{r+o(r)}{r}\log\left[  1-\left(\sqrt{1-\widetilde{p}_{-1}}-\sqrt{\widetilde{p}_{-1}}\right)^2\right]\\
    &=\log\left[  1-\left(\sqrt{1-\widetilde{p}_{-1}}-\sqrt{\widetilde{p}_{-1}}\right)^2\right]>\log\left(\frac{4}{5}\right).
\end{align*}
This proves the result.
\end{proof} 

The result presented in Theorem \ref{Bound} is particularly useful, as it provides the value of the optimal width of the main cluster. Indeed, it becomes evident that, in a voting round with a sufficiently large number of candidates $n$, choosing $l_{n,M(n)}=\displaystyle\frac{1}{5}\, n$ will increase the winning probability for the preferred candidate. In addition, this can be achieved without involving an excessively large number of voters $m=M(n)$, as condition \eqref{eq:limMlog} guarantees that the result holds as long as $M(n)$ grows asymptotically faster than $\log n$. It is worth emphasizing that it is not necessary for the number of voters $m$ to be greater than or equal to $n$ to attain the preferred candidate's success.

\section{A Continuous Limit of the Discrete Model}  \label{continuo}

While Theorem \ref{Bound} characterizes the asymptotic regime $M(n) \gg \log n$, we now examine the system under the low-density condition $M(n) = o(\sqrt{n})$. To this end, we introduce a continuous framework that corresponds---in a precise probabilistic sense detailed below---to the exact limit of the discrete model as $n \to \infty$.

\subsection{Formalization of the Continuous Model}

Let $U_1, U_2, \dots, U_m$ be independent and identically distributed (i.i.d.) random variables strictly uniformly distributed on $[0, 1)$. For a given width $\eta \in (0, 1/2)$, we define two disjoint open intervals representing the continuous counterparts of the discrete recruitment clusters:
\begin{itemize}
    \item $\tilde{B} = (1-\eta, 1)$,
    \item $\tilde{A} = (0, \eta)$.
\end{itemize}

For any generic open arc $(a,b)$ modulo 1 on the circular domain, let $N_{(a,b)}$ denote the random variable counting the number of points falling into this interval:
\begin{equation}
    N_{(a,b)} = \sum_{j=1}^m \mathbf{I}_{(a,b)}(U_j).
\end{equation}
By definition, the random vector $(N_{\tilde{B}}, N_{\tilde{A}}, m - N_{\tilde{B}} - N_{\tilde{A}})$ follows a multinomial distribution with parameters $m$ and $(\eta, \eta, 1-2\eta)$.

Since the probability of any $U_j$ falling exactly on the boundaries is zero, we can neglect boundary ties almost surely. Thus, we define the continuous decisive event $F_{\infty, m, \eta}$ simply as the simultaneous occurrence of the following conditions:
\begin{equation} \label{eq:F_infinity_conditions}
    N_{(1-\eta, 1)} \ge 2, \quad N_{(0, \eta)} \ge 2, \quad \text{and} \quad N_{(0, \eta)} < \frac{m}{2}.
\end{equation}

This event arises naturally as the continuous limit of the rules governing the discrete voting process. Specifically, it identifies the winning threshold for candidates associated with non-negligible recruitment clusters, effectively precluding ties---scenarios that, in the discrete model, would arbitrarily interrupt the voting process.

The probability of victory $p(\infty, m, \eta) := \mathsf{P}(F_{\infty,m, \eta})$ for the target candidate can be directly computed. By expressing this probability through the multinomial mass function and setting $s = N_{(0, \eta)}$ and $t = N_{(1-\eta, 1)}$, we obtain:
\begin{equation}\label{sommasemplice}
    p(\infty, m, \eta) = \sum_{s=2}^{\lfloor \frac{m-1}{2} \rfloor} \sum_{t=2}^{m-s} \frac{m!}{s! \, t! \, (m - s - t)!} \, \eta^{s} \, \eta^{t} \, (1 - 2\eta)^{m - s - t}.
\end{equation}
The summations iterate over all realizations $(s, t)$ satisfying the conditions defined in \eqref{eq:F_infinity_conditions}, ensuring that both the target candidate and the primary challenger receive at least 2 votes, while strictly preventing the challenger from reaching the winning threshold of $m/2$.

This polynomial expression in $\eta$ establishes a direct link between the geometric configuration of the continuous intervals and the electoral outcome. As we establish in Theorem \ref{T3}, these constraints are necessary and sufficient to characterize the asymptotic behavior of the system. 

However, justifying this continuous approximation requires proving that the discrete model behaves regularly as $n \to \infty$. Specifically, we must show that the probability of structural anomalies---such as multiple voters colliding on the exact same candidate---vanishes. Returning to the discrete model with $n$ candidates, $m$ voters, and $R \in O_{m,n}$, we define the collision-free event:
\begin{equation} \label{H_n}
    H_{m,n} = \left( \bigcap_{i \in A \setminus \{1\}} \{ \Xi_A(i) \leq 1 \} \right) \cap \left( \bigcap_{j \in B \setminus \{l+1\}} \{ \Xi_B(j) \leq 1 \} \right).
\end{equation}

The following lemma demonstrates that the probability of $H_{m,n}$ tends to $1$ as $n \to \infty$. This result follows directly from the generalized birthday problem \cite{DasGupta}; a brief proof is provided for completeness.

\begin{lemma}[Asymptotic Absence of Voter Collisions] \label{compleanni}
Let $n \in \mathbb{N}$ denote the number of candidates and $m = m(n)$ the number of voters. Under Assumption \ref{assumption_1}, if $m(n) = o(\sqrt{n})$, then 
\begin{equation*}
    \lim_{n \to \infty} \mathsf{P}(H_{m(n), n}) = 1.
\end{equation*}
\end{lemma}

\begin{proof}
Let $V_i := C^-_A(i) \cup \{i\}$ denote the voter selection set for candidate $i \in A \setminus \{1\}$. By definition, $|V_i| \leq 2$. For each $i \in A \setminus \{1\}$, the probability that $\Xi_A(i) > 1$ is bounded by the union bound over all pairs of voters:
\begin{equation}
    \mathsf{P}(\Xi_A(i) > 1) \leq \binom{m}{2} \left( \frac{|V_i|}{n} \right)^2 \leq \frac{m(m-1)}{2} \cdot \frac{4}{n^2} \leq \frac{2m^2}{n^2}.
\end{equation}
Taking the union bound over all $i \in A \setminus \{1\}$, where $|A \setminus \{1\}| = n/2 - 1$, the probability of a collision in $A$ is bounded by:
\begin{equation}
    \mathsf{P}\left( \bigcup_{i \in A \setminus \{1\}} \{ \Xi_A(i) > 1 \} \right) \leq \frac{n}{2} \cdot \frac{2m^2}{n^2} = \frac{m^2}{n}.
\end{equation}
By symmetry, the exact same bound applies to the set $B \setminus \{l+1\}$:
\begin{equation}
    \mathsf{P}\left( \bigcup_{j \in B \setminus \{l+1\}} \{ \Xi_B(j) > 1 \} \right) \leq \frac{m^2}{n}.
\end{equation}
Applying the union bound to the complement event $H_{m,n}^c$:
\begin{equation}
    \mathsf{P}(H_{m,n}^c) \leq \mathsf{P}\left( \bigcup_{i \in A \setminus \{1\}} \{ \Xi_A(i) > 1 \} \right) + \mathsf{P}\left( \bigcup_{j \in B \setminus \{l+1\}} \{ \Xi_B(j) > 1 \} \right) \leq \frac{2m^2}{n}.
\end{equation}
Since $m(n) = o(\sqrt{n})$, it follows that $2m(n)^2/n \to 0$ as $n \to \infty$. This confirms that $\mathsf{P}(H_{m(n), n}) \to 1$, completing the proof.
\end{proof}

The result above ensures that, as $n \to \infty$, the likelihood of local collisions or adjacency interferences vanishes asymptotically. This absence of collisions allows us to establish that the discrete winning probability converges, in the limit, to the probability defined by the continuous model, as formalized in the following theorem.

\begin{theorem} \label{T3}
Let $\{m(n)\}_{n \in \mathbb{N}}$ be a sequence of the number of voters such that $m(n) = o(\sqrt{n})$. Let $\eta \in \left(0, \frac{1}{2}\right)$ be fixed, and let $l_n$ be a sequence of integers representing the discrete cluster size such that $l_n = \eta n + o(\sqrt{n})$. The following implications hold:
\begin{enumerate}
    \item[(i)] If $\lim_{n \to \infty} m(n) = m_0 \in \mathbb{N}$, then 
    $
    \lim_{n \to \infty} p_1 (n, m(n), l_n) = p(\infty, m_0, \eta).
    $
    \item[(ii)] If $\lim_{n \to \infty} m(n) = \infty$, then 
    $
    \lim_{n \to \infty} p_1 (n, m(n), l_n) = \lim_{m \to \infty} p(\infty, m, \eta) = 1.
    $
\end{enumerate}
\end{theorem}
\begin{proof}
The convergence relies on a coupling argument defined on a common probability space. Let $U_1, U_2, \dots$ be a sequence of independent and identically distributed (i.i.d.) random variables following a uniform distribution on $[0, 1)$. We couple the continuous preferences with the discrete seeds by mapping $s^{(n)}_j = \lceil n U_j \rceil$ for each voter $j \in \{1, \dots, m(n)\}$ and for any number of candidates $n \in \mathbb{N}$.

\textit{Proof of (i).} Assume $\lim_{n \to \infty} m(n) = m_0 \in \mathbb{N}$. Since $m(n)$ is a sequence of integers, there exists $n_0 \in \mathbb{N}$ such that $m(n) = m_0$ for all $n \geq n_0$. Let $\mathbf{U} = (U_1, \dots, U_{m_0})$.

Under this coupling, the discrete winning event $F_{1,n, m_0, l_n}$ and the continuous event $F_{\infty, m_0, \eta}$ can differ only if either a collision occurs in the discrete model (meaning the event $H_{m_0,n}$ does not hold) or at least one variable $U_j$ falls into a region where the continuous and discrete classifications do not coincide. Given the mapping $s^{(n)}_j = \lceil n U_j \rceil$, this misclassification occurs if $U_j$ falls in the symmetric difference between the continuous cluster intervals (defined by $\eta$) and their discrete counterparts (defined by the normalized boundary $l_n/n$). Thus, the discrepancy is induced by the boundary shift $l_n/n - \eta$ combined with the discretization error, which has a maximum size of $1/n$. Consequently, the difference between the discrete winning event $F_{1,n, m_0, l_n}$ and the continuous event $F_{\infty, m_0, \eta}$ is contained within the union of the collision event $H_{m_0,n}^c$ and the event where at least one variable $U_j$ falls into these discretization-sensitive regions. Let $\epsilon_n = \left| \frac{l_n}{n} - \eta \right| + \frac{1}{n}$ represent the maximum error margin, given by the sum of the absolute boundary shift and the maximum discretization error. We define the boundary zones as the following neighborhoods around the critical points:
\begin{equation*}
\mathcal{B}_n = \bigcup_{x \in \{0, \eta, 1-\eta\}} (x - \epsilon_n, x + \epsilon_n) \pmod 1.
\end{equation*}

If $U_j \notin \mathcal{B}_n$ for all $j=1, \dots, m_0$ and the event $H_{m_0, n}$ occurs, the two events coincide. Therefore, applying the union bound:
\begin{align*}
\left| \mathsf{P}(F_{1,n, m_0, l_n}) - \mathsf{P}(F_{\infty, m_0, \eta}) \right| &\leq \mathsf{P} \left( H_{m_0,n}^c \cup \bigcup_{j=1}^{m_0} \{U_j \in \mathcal{B}_n\} \right) \\
&\leq \mathsf{P}(H_{m_0,n}^c) + \sum_{j=1}^{m_0} \mathsf{P}(U_j \in \mathcal{B}_n) \\
&\leq \frac{2m_0^2}{n} + m_0 \cdot \lambda(\mathcal{B}_n) \\
&\leq \frac{2m_0^2}{n} + 6 m_0 \left( \left| \frac{l_n}{n} - \eta \right| + \frac{1}{n} \right),
\end{align*}
where $\lambda(\mathcal{B}_n) \leq 3 \cdot 2\epsilon_n = 6\epsilon_n$ is the Lebesgue measure of the three boundary neighborhoods.

Since $m_0$ is constant, $\lim_{n \to \infty} \frac{1}{n} = 0$, and $\lim_{n \to \infty} \frac{l_n}{n} = \eta$, the right-hand side of the inequality explicitly vanishes as $n \to \infty$. This proves that:
$$
\lim_{n \to \infty} p_1(n, m(n), l_n) = p(\infty, m_0, \eta).
$$

\textit{Proof of (ii).}
 Assume $\lim_{n \to \infty} m(n) = \infty$ with the growth rate constraint $m(n) = o(\sqrt{n})$. Recall that $l_n = \eta n + o(\sqrt{n})$, which implies that the boundary shift is $\left| \frac{l_n}{n} - \eta \right| = o(n^{-1/2})$.

We apply the exact same coupling bound derived in Part 1, substituting the constant $m_0$ with the sequence $m(n)$:
$$
\left| p_1(n, m(n), l_n) - p(\infty, m(n), \eta) \right| \leq \frac{2m(n)^2}{n} + 6 m(n) \left( \left| \frac{l_n}{n} - \eta \right| + \frac{1}{n} \right).
$$

Let us analyze the asymptotic behavior of the right-hand side as $n \to \infty$:
\begin{itemize}
    \item Since $m(n) = o(\sqrt{n})$, the collision term $\frac{2m(n)^2}{n} \to 0$.
    \item The discretization error term $\frac{6m(n)}{n} \to 0$.
\item By the boundary shift condition, the cross term satisfies
    $$
    6 m(n) \left| \frac{l_n}{n} - \eta \right| = o(\sqrt{n}) \cdot o(n^{-1/2}) = o(1).
    $$
\end{itemize}

Therefore, the distance between the discrete and continuous probabilities asymptotically
 vanishes:
\begin{equation}\label{differenza}
\lim_{n \to \infty} \left| p_1(n, m(n), l_n) - p(\infty, m(n), \eta) \right| = 0.
\end{equation}

Finally, we analyze the continuous probability $p(\infty, m(n), \eta)$ as $m(n) \to \infty$. In the continuous limit, the votes are assigned to the recruitment clusters independently with fixed probability $\eta$. By the Law of Large Numbers, the proportion of uniform random variables falling into the interval $\tilde{A}$ converges to its measure $\eta < 1/2$. Thus, the probability that the required quotas are met approaches $1$ as the number of variables $m(n)$ goes to infinity. Thus
\begin{equation}\label{a1}
\lim_{m \to \infty} p(\infty, m, \eta) = 1.
\end{equation}

Combining the two limits \eqref{differenza} and \eqref{a1}, we obtain the desired result:
$$
\lim_{n \to \infty} p_1(n, m(n), l_n) = 1.
$$
\end{proof}

The result of Theorem \ref{T3} establishes the asymptotic equivalence between the discrete and continuous models in the regime where the number of candidates $n$ grows much faster than the number of voters $m$ (specifically, $m = o(\sqrt{n})$). On the other hand, the results of the previous sections (e.g., Theorem \ref{limsup_1} and Theorem \ref{Bound}) govern the regime where $m$ grows significantly faster than $\log n$. 

Since these two regimes overlap, combining these complementary approaches allows us to establish a universal convergence result. The following corollary demonstrates that the probability of the preferred candidate winning converges to $1$ whenever the voter population $m$ goes to infinity, regardless of the specific behavior of the sequence of candidates $n$.
\begin{corollary}[Universal Convergence to Victory] \label{Cor:uniform}
    Let $[\eta_{\min}, \eta_{\max}] \subset \left(0, \frac{1}{2}\right)$ be a fixed compact interval. Let $(m_k)_{k \in \mathbb{N}}$ be a sequence of voters such that $\lim_{k \to \infty} m_k = \infty$, and let $(n_k)_{k \in \mathbb{N}}$ be an arbitrary sequence of candidates ($n_k \geq 6$). 
    
    If $(l_k)_{k \in \mathbb{N}}$ is a sequence of discrete cluster sizes such that their relative width satisfies 
    $$
    \frac{l_k}{n_k} \in [\eta_{\min}, \eta_{\max}] \quad \text{for all } k,
    $$
    then the probability of the preferred candidate winning converges to $1$. That is,
    $$
    \lim_{k \to \infty} p_1(n_k, m_k, l_k) = 1.
    $$
\end{corollary}

\begin{proof}
    To prove the result for an arbitrary sequence of pairs $\big( (n_k, m_k) \big)_{k \in \mathbb{N}}$, we partition the sequence into two complementary subsequences based on their relative growth rates. We fix a threshold exponent, for instance $\alpha = 1/3$, and define:
    \begin{itemize}
        \item \textit{Regime 1:} The subsequence of indices $k$ where $m_k \leq n_k^{1/3}$.
        \item \textit{Regime 2:} The subsequence of indices $k$ where $m_k > n_k^{1/3}$.
    \end{itemize}
    
    \textit{Analysis of Regime 1.} For this subsequence, since $m_k \leq n_k^{1/3}$, it strictly holds that $m_k = o(\sqrt{n_k})$. We couple the $k$-th discrete model directly with a continuous model having boundary $\eta_k = l_k/n_k$. Because $\eta_k$ matches the normalized discrete boundary exactly, the boundary shift is identically zero. Following the coupling argument from Theorem \ref{T3}, the error between the probabilities is bounded solely by the collision and discretization terms:
    $$
    \left| p_1(n_k, m_k, l_k) - p(\infty, m_k, \eta_k) \right| \leq \frac{2m_k^2}{n_k} + \frac{6m_k}{n_k}.
    $$
    Since $m_k \leq n_k^{1/3}$, this coupling error strictly vanishes as $k \to \infty$. To show that the continuous probability $p(\infty, m_k, \eta_k)$ approaches $1$, we analyze the complement of the winning event. The continuous event fails if $N_{\tilde{B}} < 2$, $N_{\tilde{A}} < 2$, or $N_{\tilde{A}} \geq m_k/2$. Since $\eta_k \in [\eta_{\min}, \eta_{\max}] \subset \left(0, \frac{1}{2}\right)$, the probabilities of these failure modes are uniformly bounded: the probability of any candidate receiving fewer than two votes is bounded by the case $\eta = \eta_{\min} > 0$, while the probability of the challenger reaching the threshold ($N_{\tilde{A}} \geq m_k/2$) is bounded by the case $\eta = \eta_{\max} < 1/2$. As $m_k \to \infty$, these uniform upper bounds vanish 
exponentially. Consequently, $p(\infty, m_k, \eta_k) \to 1$, which implies that $p_1(n_k, m_k, l_k) \to 1$ along this subsequence.
    
    \textit{Analysis of Regime 2.} For this subsequence, we have $m_k > n_k^{1/3}$. This implies $n_k < m_k^3$, and therefore:
    $$
    \lim_{k \to \infty} \frac{\log(n_k - 1)}{m_k} \leq \lim_{k \to \infty} \frac{3 \log m_k}{m_k} = 0.
    $$
    This condition allows us to apply the Chernoff upper bound established in Equation \eqref{diseq_1}. For the failure probability, substituting $n = n_k$, $m = m_k$, and $l = l_k$, we have:
    \begin{multline*}
        1 - p_1(n_k, m_k, l_k) \leq \exp \Biggl\{ m_k \Biggl[ \frac{\log(n_k-1)}{m_k} \\
        + \log\left( \max \left\{ 1-\frac{\big( \sqrt{l_k+1}-\sqrt{2}\big)^2}{n_k}, \, 1-\frac{\big( \sqrt{n_k-l_k}-\sqrt{l_k}\big)^2}{n_k} \right\} \right) \Biggr] \Biggr\}.
    \end{multline*}
    As $k \to \infty$, the arguments inside the maximum are asymptotically equivalent to $1 - \frac{l_k}{n_k}$ and $2\sqrt{\frac{l_k}{n_k} \left(1 - \frac{l_k}{n_k}\right)}$. By hypothesis, the ratio $\frac{l_k}{n_k}$ is strictly confined within the compact interval $[\eta_{\min}, \eta_{\max}] \subset \left(0, \frac{1}{2}\right)$. Over this compact domain, both of these continuous functions are strictly bounded above by a value less than $1$ (since $\eta_{\min} > 0$ ensures $1-\eta < 1$, and $\eta_{\max} < 1/2$ ensures $2\sqrt{\eta(1-\eta)} < 1$).  Consequently, their maximum is also strictly bounded away from $1$, meaning its logarithm is bounded above by some strictly negative constant $c^* < 0$. 
    
    Since $\frac{\log(n_k-1)}{m_k} \to 0$, the argument of the exponential is asymptotically dominated by this strictly negative constant. Thus, the failure probability is bounded by $e^{m_k c^*}$, which exponentially vanishes as $m_k \to \infty$, implying $p_1(n_k, m_k, l_k) \to 1$.

    Since the probability $p_1(n_k, m_k, l_k)$ converges to $1$ along both complementary subsequences, the limit $\lim_{k \to \infty} p_1(n_k, m_k, l_k) = 1$ holds for the entire sequence.
\end{proof}

\begin{rem}[Boundary Behavior of the Success Probability]
The restriction of the relative cluster width to a compact interval $[\eta_{\min}, \eta_{\max}] \subset \left(0, \frac{1}{2}\right)$ in Corollary \ref{Cor:uniform} is not merely a technical artifact needed for uniform convergence, but reflects the intrinsic limitations of the voting mechanism at the extremes. 

Indeed, examining the continuous probability as $\eta$ approaches the boundaries clarifies why successful configurations must avoid them:
\begin{itemize}
    \item As $\eta \to 0$, the target candidate's coverage area shrinks to zero. Consequently, they receive almost no votes, leading to $\lim_{\eta \to 0} p(\infty, m, \eta) = 0$.
    \item As $\eta \to 1/2$, the intervals for the target candidate and the potential challenger evenly partition the entire circle. The election degenerates into a symmetric two-way race where voters act as independent coin tosses. Thus, the probability of the target candidate strictly winning is bounded by symmetry, yielding $\lim_{\eta \to 1/2} p(\infty, m, \eta) \leq \frac{1}{2}$.
\end{itemize}
Therefore, achieving a winning probability that converges to $1$ fundamentally requires the parameter $\eta$ to be strictly separated from both $0$ and $1/2$.
\end{rem}

Having established the asymptotic equivalence between the discrete 
and continuous success probabilities, a natural question arises from a 
mechanism design perspective: how should the election organizer 
choose the discrete cluster size $l$ to maximize the target candidate's chances? 

The following theorem demonstrates that optimizing the discrete model is asymptotically equivalent to optimizing its continuous counterpart. Specifically, for a fixed number 
of voters $m$, as the pool of candidates $n$ grows arbitrarily large, the optimal 
normalized width of the discrete cluster naturally converges to the 
optimal continuous parameter $\eta$.

\begin{theorem}[Asymptotic Convergence of the Optimal Cluster Width] \label{Teo:ottimo_continuo}
Let $m \in \mathbb{N}$ be a fixed number of voters. Let $\mathcal{E}^*(m) \subset \left(0, \frac{1}{2}\right)$ be the set of global maximizers for the continuous success probability:
$$
\mathcal{E}^*(m) = \underset{\eta \in (0, 1/2)}{\arg\max} \, p(\infty, m, \eta).
$$
For each $n \in \mathbb{N}$, let $l_{opt}(n, m) \in \underset{l}{\arg\max} \, p_1(n, m, l)$ be an optimal discrete cluster width. 
Then, as the number of candidates $n \to \infty$, the sequence of normalized optimal widths approaches the set of continuous maximizers:
$$
\lim_{n \to \infty} \inf_{\eta^* \in \mathcal{E}^*(m)} \left| \frac{l_{opt}(n, m)}{n} - \eta^* \right| = 0.
$$
Furthermore, if the continuous probability $p(\infty, m, \eta)$ admits a unique global maximum $\eta_{opt}(m)$ over $\left(0, \frac{1}{2}\right)$, the sequence strongly converges:
$$
\lim_{n \to \infty} \frac{l_{opt}(n, m)}{n} = \eta_{opt}(m).
$$
\end{theorem}

\begin{proof}
By the coupling argument established in Theorem \ref{T3}, for a fixed number of voters $m$, the absolute difference between the discrete and continuous probabilities satisfies:
$$
\left| p_1(n, m, l) - p\left(\infty, m, \frac{l}{n}\right) \right| \leq \frac{2m^2}{n} + \frac{6m}{n}.
$$
Crucially, this upper bound depends only on $m$ and $n$, and is completely independent of $l$. Therefore, as $n \to \infty$, the discrete probability uniformly converges to the continuous probability across the entire domain $l/n \in \left[0, \frac{1}{2}\right]$. 

Let $x_n = \frac{l_{opt}(n, m)}{n}$ be the sequence of normalized optimal discrete widths, where we suppress the explicit dependence on the fixed parameter $m$ for notational simplicity. Since $x_n \in \left[0, \frac{1}{2}\right]$ for all $n$, by the Bolzano-Weierstrass theorem, the sequence must have at least one limit point. 

Let $\bar{\eta}$ be any limit point of the sequence $(x_n)_{n \in \mathbb{N}}$, and let $(x_{n_k})_{k \in \mathbb{N}}$ be a subsequence converging to $\bar{\eta}$. We want to show that $\bar{\eta} \in \mathcal{E}^*(m)$. 
Let $\eta^* \in \mathcal{E}^*(m)$ be a true continuous global maximizer. By the definition of $l_{opt}$ as the discrete maximizer, for any $k$ we have the inequality: 
$$
p_1(n_k, m, x_{n_k} n_k) \geq p_1(n_k, m, \lfloor \eta^* n_k \rfloor).
$$
We now take the limit as $k \to \infty$ on both sides. For the right-hand side, since $\frac{\lfloor \eta^* n_k \rfloor}{n_k} \to \eta^*$, the uniform convergence yields $\lim_{k \to \infty} p_1(n_k, m, \lfloor \eta^* n_k \rfloor) = p(\infty, m, \eta^*)$. 
For the left-hand side, by the uniform convergence and the continuity of the continuous probability function $p(\infty, m, \cdot)$, we obtain $\lim_{k \to \infty} p_1(n_k, m, x_{n_k} n_k) = p(\infty, m, \bar{\eta})$. 

Preserving the inequality in the limit yields:
$$
p(\infty, m, \bar{\eta}) \geq p(\infty, m, \eta^*).
$$
Since $\eta^*$ is, by definition, a global maximum of $p(\infty, m, \cdot)$, it must also hold that $p(\infty, m, \eta^*) \geq p(\infty, m, \bar{\eta})$. Therefore, $p(\infty, m, \bar{\eta}) = p(\infty, m, \eta^*)$, which implies $\bar{\eta} \in \mathcal{E}^*(m)$.
Thus, all limit points of $x_n$ belong to the optimal set $\mathcal{E}^*(m)$, proving the first statement.

The second statement follows immediately from the first: if $\mathcal{E}^*(m)$ is a singleton $\{\eta_{opt}(m)\}$, the infimum simplifies to $\left| x_n - \eta_{opt}(m) \right|$. 
Since the limit of this distance is $0$, the sequence strongly converges to $\eta_{opt}(m)$.
\end{proof}

\begin{rem} \label{rem:unique_max}
While an analytical proof of the uniqueness of the global maximum $\eta_{opt}(m)$ for the polynomial $p(\infty, m, \eta)$ remains elusive due to the algebraic complexity of the multinomial sums, numerical evaluations consistently demonstrate a single, well-defined peak in the interval $\left(0, \frac{1}{2}\right)$ for all tested values of $m$. Consequently, assuming this uniqueness, the sequence of discrete optimal cluster proportions converges to this unique continuous maximum.
\end{rem}

\subsection{Asymptotic Behavior of the Universal Victory Event}

The results established so far characterize the asymptotic behavior of $p_1(n, m, l)$, the success probability for a specific target candidate under a fixed initial partition. However, as introduced in Section \ref{main}, a significantly stronger condition from a mechanism design perspective is the \textit{Universal Victory Event}, denoted as $F^{(2)}_{n,m,l}$. 

Recall that $F^{(2)}_{n,m,l} = \bigcap_{i=1}^n F_{i,n,m,l}$ represents the event where a single, fixed voting profile $R \in O_{m,n}$ guarantees that \textit{every} candidate $i \in [n]$ can be made to win, provided the initial sets $A^{(i,n,l)}$ and $B^{(i,n,l)}$ are chosen according to the corresponding rotational shift. 

To extend our asymptotic analysis to this universal event, we first define its continuous counterpart. Let $\theta \in [0, 1)$ represent an arbitrary continuous rotational shift of the cluster. The continuous universal event requires the victory conditions to hold simultaneously for all possible rotations. Let $F_{\theta, \infty, m, \eta}$ be the continuous victory event when the target candidate is shifted by $\theta$. We define $F^{(2)}_{\infty, m, \eta}$ as:
\begin{equation}\label{F(2)}
F^{(2)}_{\infty, m, \eta} := \bigcap_{\theta \in [0, 1)} F_{\theta, \infty, m, \eta} = \bigcap_{\theta \in [0, 1)} \left\{ N_{(\theta, \theta+\eta)} \ge 2, \quad N_{(\theta-\eta, \theta)} \ge 2, \quad \text{and} \quad
 N_{(\theta, \theta+\eta)} < \frac{m}{2} \right\}.
\end{equation}
Notice that by imposing $N_{(\theta, \theta+\eta)} \ge 2$ for all $\theta \in [0,1)$, the requirement for the adjacent backward interval is inherently satisfied by symmetry. This allows us to express the continuous universal event elegantly in terms of the global minimum and maximum voter counts over any sliding window of length $\eta$
\begin{equation} \label{eq:F2_continuous_simplified}
F^{(2)}_{\infty, m, \eta} = \left\{ \min_{\theta \in [0, 1)} N_{(\theta, \theta+\eta)} \ge 2 \quad \text{and} \quad \max_{\theta \in [0, 1)} N_{(\theta, \theta+\eta)} < \frac{m}{2} \right\}.
\end{equation}

Since the universal event $F^{(2)}$ implies the standard resolution event $F_1$, it imposes stricter structural constraints on the uniformity of the underlying voter distribution. Nevertheless, the probability of this universal event exhibits similar asymptotic convergence properties.

\begin{rem}[Universal Victory in Randomized Mechanisms] \label{rem:debaets_desantis}
It is worth noting that universal asymptotic behaviors of a similar nature have been investigated in other randomized voting frameworks. For instance, in the context of single-elimination tournaments (see De Baets and De Santis \cite{desantisdebaets}), it has been shown that a structural manipulation of the tournament bracket can guarantee the victory of \textit{any} target candidate, provided the number of voters $m$ grows at a specific rate relative to the number of candidates $n$ (specifically, requiring $m$ to grow as $(\log n)^3$). While the underlying mechanics and the structural rules of our spatial model differ, both approaches formally illustrate how a sufficiently balanced voter distribution allows the mechanism designer to universally determine the outcome by adjusting the initial configuration (either the bracket seeding or the geometric partition) as the electorate size expands.
\end{rem}

\begin{proposition}[Asymptotic Behavior of the Universal Event] \label{prop:robust_victory}
Under the same assumptions as Theorem \ref{T3} and Corollary \ref{Cor:uniform}, the probability of the universal event satisfies:
\begin{enumerate}
    \item[(i)] \textit{Pathwise discrete-continuous coupling:} For any fixed $m$, $\lim_{n \to \infty} p^{(2)}(n, m, l_n) = p^{(2)}(\infty, m, \eta)$.
    \item[(ii)] \textit{Universal sure victory:} $\lim_{m \to \infty} p^{(2)}(\infty, m, \eta) = 1$. Consequently, as $m \to \infty$, the discrete universal probability converges to $1$.
\end{enumerate}
\end{proposition}
\begin{proof}
We highlight the differences from the proofs of Theorem \ref{T3} and Corollary \ref{Cor:uniform}, focusing on the uniformity over the rotational shift $\theta$.

\textit{Proof of (i).} In the fixed $m$ regime, the continuous-to-discrete coupling must hold simultaneously for all $\theta \in [0,1)$. Since the $m$ voters $U_1, \dots, U_m$ are i.i.d. uniform random variables on the unit circle, almost surely no two voters coincide, and no two voters are separated by an arc of exact length $\eta$ or $1-\eta$. Let $d(x,y) = \min(|x-y|, 1-|x-y|)$ denote the circular distance. We can define a minimum strictly positive gap $\delta > 0$ between any voter $U_i$ and the critical boundaries generated by any other voter $U_j$:
$$
\delta = \min_{i \neq j} \min \Big\{ d(U_i, U_j), \big| d(U_i, U_j) - \eta \big| \Big\} > 0 \quad \text{a.s.}
$$
The maximum boundary shift between the discrete and continuous grids is $\epsilon_n = \left| \frac{l_n}{n} - \eta \right| + \frac{1}{n}$, which vanishes as $n \to \infty$. 
Thus, there almost surely exists an $N_0$ such that for all $n > N_0$, $\epsilon_n < \delta/2$. Once this threshold is crossed, the discrete grid resolution is fine enough that the discrete and continuous boundaries cannot "jump over" any voter $U_j$, regardless of the rotation $\theta$. This implies that the continuous and discrete universal events coincide almost surely (pathwise) for sufficiently large $n$, proving the first statement.

\textit{Proof of (ii).} As $m \to \infty$, we evaluate the continuous event. By the Glivenko-Cantelli theorem, the empirical cumulative distribution function of the uniform voters converges uniformly to the true distribution (which is $F(x)=x$ for the uniform distribution on $[0,1)$). This implies that the empirical measure of \textit{any} arc of length $\eta$ converges uniformly to $\eta$:
$$
\lim_{m \to \infty} \sup_{\theta \in [0, 1)} \left| \frac{1}{m} N_{(\theta, \theta+\eta)} - \eta \right| = 0 \quad \text{a.s.}
$$
Since $\eta \in \left(0, \frac{1}{2}\right)$, the asymptotic maximum concentration in any such interval is strictly bounded away from $1/2$ (specifically, $m\eta < m/2$) and strictly bounded away from $0$. 
Thus, for sufficiently large $m$, the interval count $N_{(\theta, \theta+\eta)}$ will almost surely concentrate around $m\eta$. Because $m\eta$ eventually strictly exceeds $2$ and remains strictly below $m/2$, the conditions $\min_\theta N_{(\theta, \theta+\eta)} \ge 2$ and $\max_\theta N_{(\theta, \theta+\eta)} < \frac{m}{2}$ are simultaneously satisfied for all $\theta \in [0,1)$ with probability approaching $1$ as $m \to \infty$. 
\end{proof}

\begin{corollary}[Asymptotic Optimality for the Universal Event] \label{cor:ottimo_universale}
Let $m \in \mathbb{N}$ be a fixed number of voters. Let $\mathcal{E}^{(2)}_{*}(m) \subset \left(0, \frac{1}{2}\right)$ be the set of global maximizers for the continuous universal probability $p^{(2)}(\infty, m, \eta)$. For each $n \in \mathbb{N}$, let $l^{(2)}_{opt}(n, m) \in \underset{l}{\arg\max} \, p^{(2)}(n, m, l)$ be an optimal discrete cluster width for the universal event. 
Then, as $n \to \infty$, the sequence of normalized optimal widths approaches the set of continuous maximizers:
$$
\lim_{n \to \infty} \inf_{\eta^* \in \mathcal{E}^{(2)}_{*}(m)} \left| \frac{l^{(2)}_{opt}(n, m)}{n} - \eta^* \right| = 0.
$$
\end{corollary}

\begin{proof}
The result follows directly from the same arguments used in the proof of Theorem \ref{Teo:ottimo_continuo}. The pathwise coupling established in Proposition \ref{prop:robust_victory} ensures that the absolute difference between the discrete and continuous universal probabilities is uniformly bounded by a term that vanishes as $n \to \infty$, independently of the cluster width. Therefore, applying the exact same Bolzano-Weierstrass argument to the sequence of discrete universal maximizers yields the claim.
\end{proof}

\begin{rem}[On Uniqueness and Analyticity]
The assumption of a unique maximizer $\eta_{opt}(m)$ is supported by both the algebraic structure of the problem and numerical evidence. Since the probability of victory is given by a non-constant polynomial, its maxima are isolated, precluding any ``flat'' regions. Furthermore, numerical evaluations for $m \in \{11, \dots, 101\}$ consistently exhibit a strictly unimodal profile, with a single well-defined peak shifting monotonically towards $\frac{1}{5}$ as $m$ increases. This regularity ensures the well-behaved convergence of the normalized discrete optima $\frac{l_{opt}(n,m)}{n}$, confirming that the ratio $\frac{1}{5}$ is an intrinsic geometric property of the two-round cyclic mechanism, dictating the optimal strategy even for small electorates.
\end{rem}

\section{Numerical Analysis of the Target Candidate's Victory}
\label{computational}

In this section, we present a numerical analysis to validate the theoretical framework developed previously. Specifically, we investigate the winning probability of a designated candidate, $p_1(n, m, l)$, to determine the optimal strategic block size $l_{opt}(n, m)$ that maximizes this probability for a given number of candidates $n$ and voters $m$.

Our investigation covers a parameter grid where $n \in \{30, 34, \dots, 100\}$ and $m \in \{21, 25, \dots, 101\}$. We restricted the electorate size $m$ strictly to odd integers to eliminate structural parity effects; an even number of voters increases the probability of a tie, which in our model results in the defeat of all candidates, thereby artificially deflating the winning probabilities and introducing oscillations in the trends.

Preliminary exploration revealed a critical property: for any fixed $n$ and $m$, $p_1(n, m, l)$ is strictly unimodal with respect to the partition parameter $l$. To efficiently pinpoint this maximum across vastly different scales, we implemented a unified geometric approach based on the Wilson score interval: the \textit{Wilson Centroid Method}.
While for small $m$ the probability distribution $\hat{p}_1(l)$ exhibits a well-defined peak, as $m$ increases, the system undergoes a phase transition where the peak saturates into a wide "Wilson Plateau" of statistical certainty ($\hat{p}_1 \approx 1$). Instead of employing different heuristics for different regimes, the algorithm dynamically identifies the optimum using a single universal principle:

\begin{enumerate}
    \item \textit{High-Resolution Scan:} We evaluate $\hat{p}_1(n, m, l)$ across a fine grid to identify the empirical maximum probability $\hat{p}_{max}$.
    
    \item \textit{Wilson Boundary Extraction (Edge-Tracking):} Using a 95\% Wilson lower confidence bound associated with $\hat{p}_{max}$, the algorithm tracks the left and right boundaries ($l_{left}, l_{right}$) of the region that is statistically indistinguishable from the maximum. For small $m$, this region forms a tight confidence interval around the peak; for large $m$, it perfectly maps the saturated plateau.
    
    \item \textit{Geometric Centroid:} The exact discrete optimum is robustly estimated as the geometric midpoint of this stability region, maximizing the distance from failure on both sides: 
    \begin{equation}
        l_{opt} = \left\lfloor \frac{l_{left} + l_{right}}{2} \right\rceil
    \end{equation}
\end{enumerate}

To minimize variance when comparing different cluster sizes, the preliminary stages utilized Common Random Numbers (CRN), applying identical random seeds across all $l$ values for a fixed $(n, m)$ pair.
To prevent selection bias, the optimization concludes with a high-precision validation phase: an independent set of $100{,}000$ simulations is generated to evaluate the candidates and pinpoint $l_{opt}(n, m)$, ensuring an unbiased estimate of the true probability. \par
Figure \ref{fig:hybrid_regimes} illustrates this methodology applied to discrete simulation data. It contrasts a peaked regime at $m=20$ — where the plateau reduces to a narrow confidence interval near the vertex of a pseudo-parabolic peak — with a saturated regime at $m=100$, characterized by a vast stability region. Crucially, the Wilson score interval serves a dual purpose: it is actively employed within the algorithm to dynamically track the boundaries of the optimum, and it is subsequently used to establish a rigorous lower bound for the final error estimation. 

Figure \ref{fig:heatmap_prob} maps the resulting macroscopic dynamics across the parameter space, revealing two primary trends. First, for a fixed candidate pool $n$, increasing turnout $m$ drives a rapid, monotonic convergence of $p_1$ toward absolute certainty. As the voting population grows, the structural advantage of the strategic block overwhelmingly suppresses stochastic noise; for $m \ge 41$ and $n \ge 52$, the winning probability systematically exceeds $0.99$. Second, expanding the candidate pool $n$ exerts a positive, albeit more gradual, effect on $p_1$ by fragmenting non-aligned votes, thereby lowering the threshold required to survive the first round.

\begin{figure}[htbp]
    \centering
    \begin{minipage}{0.48\textwidth}
        \centering
        \includegraphics[width=\linewidth]{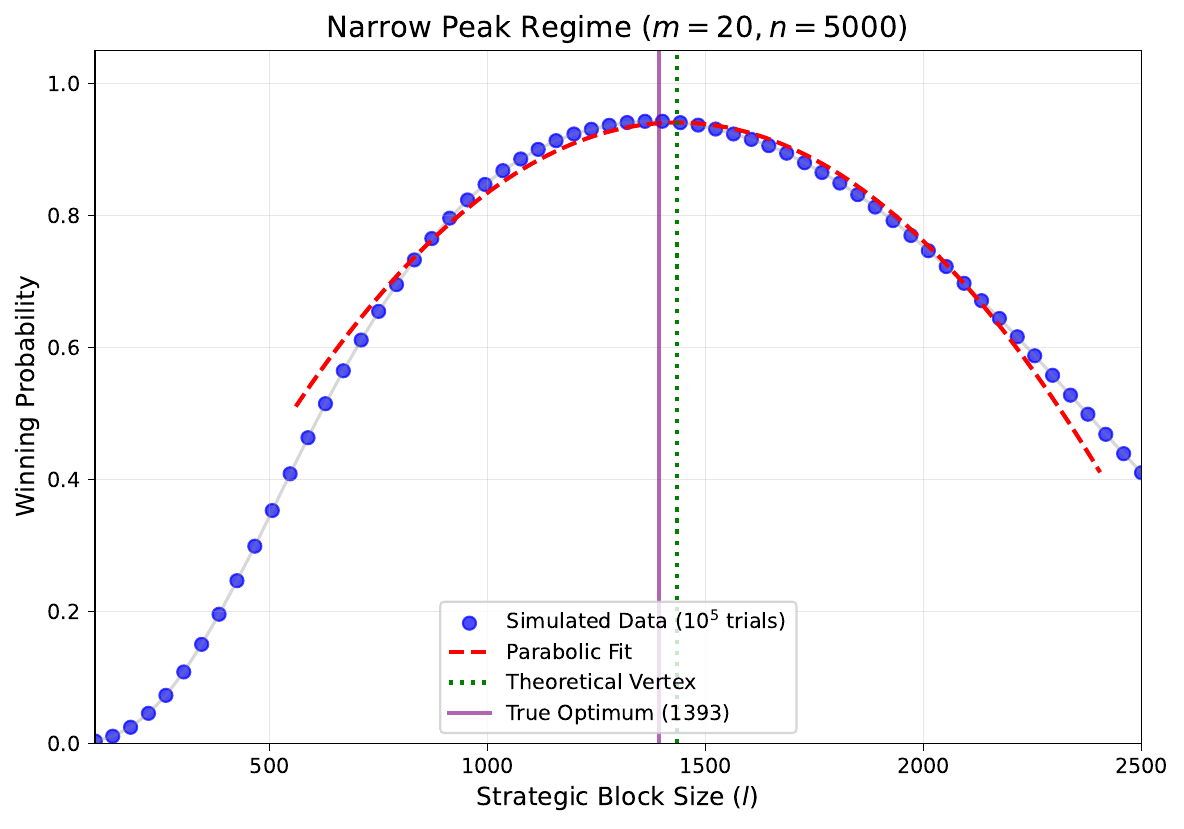}
    \end{minipage}\hfill
    \begin{minipage}{0.48\textwidth}
        \centering
        \includegraphics[width=\linewidth]{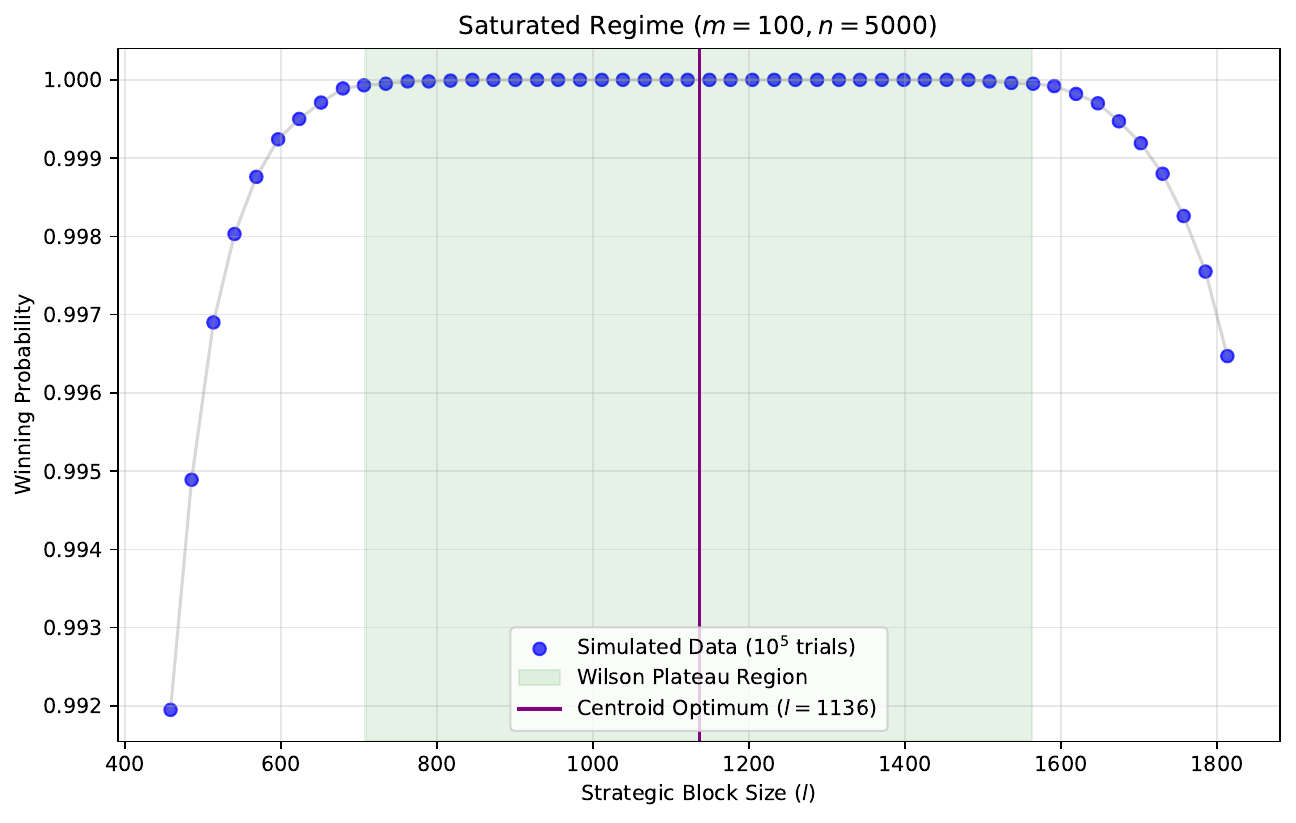}
    \end{minipage}
    \caption{Comparison of optimization regimes for $n=5000$. \textit{Left:} Peaked regime at $m=20$, where the probability follows a parabolic trend and the Wilson boundaries form a tight interval around the theoretical vertex. \textit{Right:} Saturated regime at $m=100$, where the emergence of the ``Wilson Plateau" expands the stability region across hundreds of values. In both cases, the optimal strategy $l_{opt}$ is universally identified as the centroid of this region.}
    \label{fig:hybrid_regimes}
\end{figure}

\begin{figure}[htbp]
    \centering
    \includegraphics[width=0.5\textwidth]{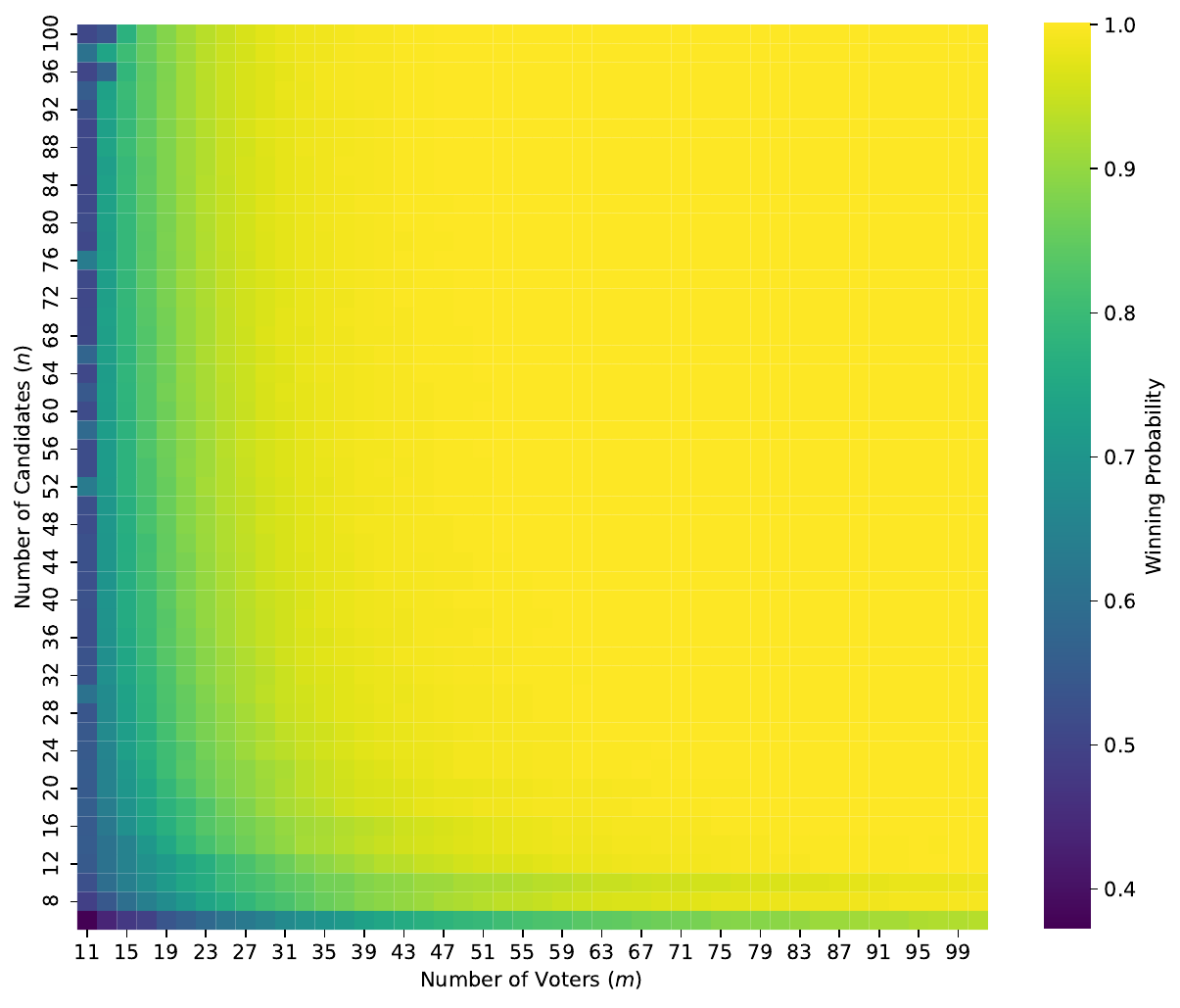}
    \caption{Heatmap of the optimized winning probability $\hat{p}_1(n, m, l_{opt}(n, m))$.}
    \label{fig:heatmap_prob}
\end{figure}

\begin{table}[htbp]
    \centering
    \renewcommand{\arraystretch}{1.2}
    \begin{tabular}{@{}ccccc@{}}
        \toprule
        \textbf{Choices} ($m$) & \textbf{Optimal Cutoff} ($l_{opt}$) & \textbf{Ratio} ($\eta = l_{opt}/n$) & \textbf{Plateau Range} & $\hat{p}_1 \pm \text{SE}_{W}$ \\
        \midrule
        \multicolumn{5}{c}{\textit{Narrow Peak Regime}} \\
        \midrule
        10 & 1698 & 0.3396 & [1618, 1777] & 0.6311 $\pm$ 0.0015 \\
        20 & 1393 & 0.2786 & [1347, 1439] & 0.9438 $\pm$ 0.0007 \\
        30 & 1297 & 0.2594 & [1240, 1354] & 0.9913 $\pm$ 0.0003 \\
        40 & 1248 & 0.2496 & [1187, 1308] & 0.9987 $\pm$ 0.0001 \\
        50 & 1213 & 0.2426 & [1123, 1303] & 0.9998 $\pm$ 0.0001 \\
        60 & 1229 & 0.2458 & [1081, 1377] & 0.9999 $\pm$ 0.0001 \\
        \midrule
        \multicolumn{5}{c}{\textit{Saturated Plateau Regime}} \\
        \midrule
        70 & 1219 & 0.2438 & [1044, 1394] & 1.0000 \\
        80 & 1140 & 0.2280 & [909, 1372] & 1.0000 \\
        90 & 1131 & 0.2262 & [801, 1461] & 1.0000 \\
        100 & 1136 & 0.2272 & [708, 1563] & 1.0000 \\
        110 & 1146 & 0.2292 & [718, 1575] & 1.0000 \\
        120 & 1111 & 0.2222 & [656, 1566] & 1.0000 \\
        130 & 1116 & 0.2232 & [582, 1650] & 1.0000 \\
        140 & 1116 & 0.2232 & [535, 1696] & 1.0000 \\
        150 & 1134 & 0.2268 & [550, 1719] & 1.0000 \\
        160 & 1104 & 0.2208 & [473, 1735] & 1.0000 \\
        \bottomrule
    \end{tabular}
    \caption{Empirical results for the continuous limit simulation with $n=5000$, based on $10^5$ independent trials per configuration. To maintain rigorous statistical consistency across all probability regimes, uncertainties ($\text{SE}_{W}$) are derived entirely from the 95\% Wilson score interval. For bounded peaks, this symmetrically aligns with standard error margins. For the saturated regime ($m \ge 70$, where $\hat{p}_1 = 1.0000$), the asymmetric Wilson interval naturally provides a strict lower confidence bound of $p_1 \ge 0.99996$, correctly modeling the exponential decay of the true failure probability without methodological switching. The uncertainty on the optimal ratio $\eta$ can be derived from the reported Plateau Range $[l_{left}, l_{right}]$ as $\Delta \eta \le (l_{right} - l_{left}) / (2n)$.}
    \label{tab:continuous_limit_results}
\end{table}

To verify the asymptotic behavior, the simulation was scaled to a candidate pool of $n = 20{,}000$ against an electorate of $m = 71$. Since increasing $n$ fragments non-aligned votes and lowers the winning threshold, this configuration yields an empirical winning probability saturated at $\hat{p}_1 = 1$ across a plateau of $l \in [4333, 4666]$. Continuing with our unified Wilson methodology, the absence of observed failures across the $10^5$ independent trials yields a 95\% Wilson upper bound for the true failure probability ($1 - p_1$) of $3.8 \times 10^{-5}$. Correspondingly, the normalized optimal block size $\eta = l_{opt}/n$ is intrinsically bounded between $0.217$ and $0.233$, confirming the convergence toward the predicted asymptotic limit.

\subsection{Continuous Limit and Empirical Validation}

To analytically ground these observations, we numerically evaluated the continuous limit established in Eq. \eqref{sommasemplice}. To prevent overflow errors from large factorials, the algorithm processes all multinomial coefficients in logarithmic space. We determined the optimal continuous density $\eta_\infty(m)$ via a grid search with step size $\Delta\eta = 0.0007$ over the interval $[0.1, 0.45]$. This bounds the numerical uncertainty on the continuous optimum to $\pm 0.00035$. By computing all combinatorial terms in logarithmic space, numerical accumulation errors in the analytic probabilities ($p_1$) are strictly negligible at the reported four-decimal precision. Table \ref{tab:comparison} compares these analytical results with our discrete simulations, providing direct empirical validation for Theorem \ref{Teo:ottimo_continuo} and Corollary \ref{Cor:uniform}.

\begin{table}[htbp]
\centering
\caption{Empirical Validation: Convergence of Discrete Optima to the Continuous Limit. For empirical probabilities reaching unity (e.g., $m=51$ at $n=10{,}000$), the absence of observed failures corresponds to a 95\% Wilson upper confidence bound of approximately $3.8 \times 10^{-5}$ for the failure rate.} 
\label{tab:comparison}
\begin{tabular}{c | cc | cc | cc}
\toprule
\multirow{2}{*}{$m$} & \multicolumn{2}{c|}{$n=98$ (Discrete)} & \multicolumn{2}{c|}{$n=10{,}000$ (Discrete)} & \multicolumn{2}{c}{Continuous Limit (Analytic)} \\
& $\eta_{98}(m)$ & $\hat{p}_1 \pm \text{SE}_W$ & $\eta_{10k}(m)$ & $\hat{p}_1 \pm \text{SE}_W$ & $\eta_\infty(m)$ & $p_1$ \\ \midrule
11 & 0.3673 & 0.6662 $\pm$ 0.0015 & 0.3488 & 0.7608 $\pm$ 0.0013 & 0.3458 & 0.7606 \\
31 & 0.3061 & 0.9787 $\pm$ 0.0005 & 0.2636 & 0.9954 $\pm$ 0.0002 & 0.2586 & 0.9959 \\
51 & 0.2755 & 0.9988 $\pm$ 0.0001 & 0.2395 & 1.0000 & 0.2385 & 0.9999 \\ \bottomrule 
\end{tabular}
\end{table}

Reading the table horizontally confirms Theorem \ref{Teo:ottimo_continuo} (Asymptotic Convergence of the Optimal Cluster): while finite-size effects exist at $n=98$, scaling the system to $n=10{,}000$ forces the empirical density $\eta_{10k}(m)$ to converge toward the continuous maximizer $\eta_\infty(m)$. Reading vertically confirms Corollary \ref{Cor:uniform} (Universal Convergence to Victory): as $m$ grows, statistical noise collapses, and the targeted coalition strategy asymptotically guarantees victory ($\lim_{m \to \infty} p_1 = 1$).

\subsection{The Law of Exponential Decay: The 4/5 Constant}

\begin{table}[htbp]
\centering
\small
\begin{tabular}{rcc|rcc}
\toprule
\textbf{$m$} & \textbf{$\eta_\infty$} & \textbf{$\log_{10}(1-p_1)$} & \textbf{$m$} & \textbf{$\eta_\infty$} & \textbf{$\log_{10}(1-p_1)$} \\ \midrule
11 & 0.3455 & $-0.62$ & 41 & 0.2460 & $-3.31$ \\
15 & 0.3100 & $-0.95$ & 45 & 0.2425 & $-3.69$ \\
21 & 0.2815 & $-1.48$ & 51 & 0.2385 & $-4.25$ \\
25 & 0.2700 & $-1.84$ & 55 & 0.2360 & $-4.63$ \\
31 & 0.2585 & $-2.39$ & 61 & 0.2330 & $-5.19$ \\
35 & 0.2525 & $-2.76$ & 69 & 0.2300 & $-5.95$ \\ \bottomrule
\end{tabular}
\caption{Exponential decay of risk in the continuous model. The failure probability scales linearly on the logarithmic axis, while $\eta_\infty$ approaches the theoretical $0.20$ limit.}
\label{tab:failure_decay}
\end{table}

The continuous model exhibits a constant decay rate for the failure probability. As demonstrated in Table \ref{tab:failure_decay}, the risk drops linearly on a logarithmic scale, yielding an asymptotic decay law of the form:
\begin{equation} \label{eq:decay_law}
    1 - p_1 \approx C \cdot \left( \frac{4}{5} \right)^m
\end{equation}
This dictates that each additional voter recruited into the process yields a constant 20\% reduction in the remaining strategic risk. By $m=69$, the failure probability falls to $10^{-5.95}$. 

This decay is linked to Theorem \ref{Teo:ottimo_continuo}, which established $\eta = 1/5$ as the optimal asymptotic partition density. The failure rate scales with the complement of this density ($1 - 1/5 = 4/5$), indicating that the institutional geometry acts as a high-pass filter. By shifting to preserve this exponential slope across varying $m$, the optimal $\eta_\infty(m)$ proves that the 4/5 constant is a structural invariant, dissipating electoral noise at a predictable rate.

\subsection{High-Precision Validation in the Extreme Asymptotic Regime}

Evaluating the continuous optimum $\eta_\infty$ requires resolving marginal probability differences as small as $10^{-30}$ during the optimization process. Because these variations fall far below standard 64-bit floating-point precision limits (machine epsilon $\approx 2.2 \times 10^{-16}$), the entire optimization routine was implemented using arbitrary-precision arithmetic via the Python \texttt{mpmath} library. This mitigates numerical flattening around the maximum and allows us to verify the exponential decay deep into the asymptotic regime, evaluating electorates up to $m = 501$.

\begin{table}[htbp]
\centering
\caption{Exact Optimal $\eta_\infty$ and Distance from Certainty via Arbitrary-Precision Arithmetic.}
\label{tab:high_precision_eta}
\begin{tabular}{|c|c|c|}
\hline
\textbf{Voters ($m$)} & \textbf{Optimal $\eta_\infty$} & \textbf{Distance from Certainty ($1 - p_1$)} \\ \hline
51  & 0.238478 & $5.61 \times 10^{-5}$ \\ \hline
151 & 0.215680 & $1.71 \times 10^{-14}$ \\ \hline
251 & 0.210277 & $4.27 \times 10^{-24}$ \\ \hline
351 & 0.207762 & $9.95 \times 10^{-34}$ \\ \hline
501 & 0.206350 & $2.65 \times 10^{-48}$ \\ \hline
\end{tabular}
\end{table}

\begin{figure}[htbp]
    \centering
    \includegraphics[width=0.65\textwidth]{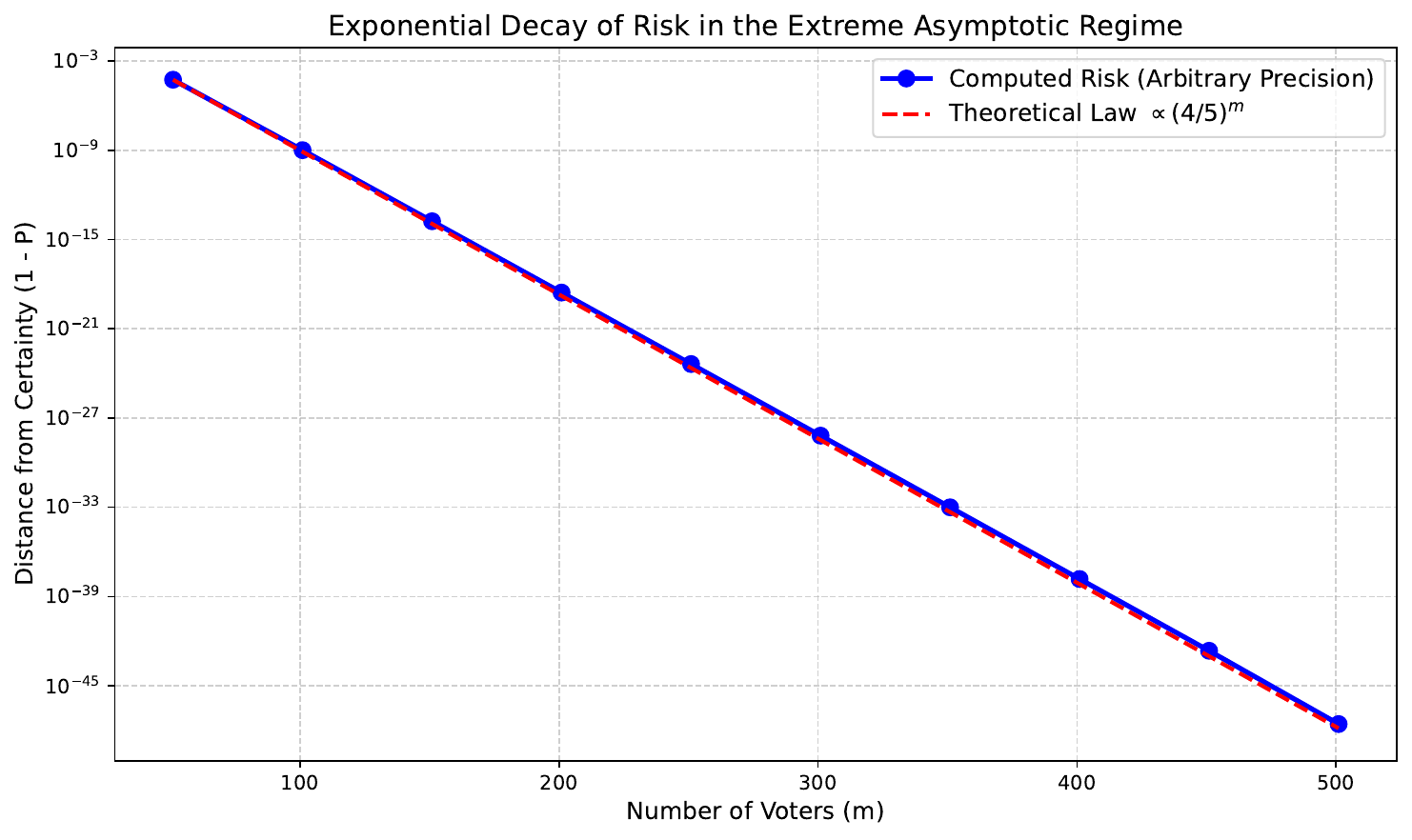}
    \caption{Semi-log plot of the distance from certainty ($1 - p_1$). The numerical values (blue dots) align with the theoretical $(4/5)^m$ law (red dashed line) across 43 orders of magnitude.}
    \label{fig:extreme_decay}
\end{figure}

As shown in Table \ref{tab:high_precision_eta} (condensed for brevity) and Figure \ref{fig:extreme_decay}, the data confirms our framework. The optimal density $\eta_\infty$ steadily approaches $0.206$ at $m=501$, converging toward the $0.20$ asymptote. Concurrently, across an increment of $\Delta m = 450$ voters, the empirical drop in risk — from $10^{-4.25}$ to exactly $10^{-47.57}$ — matches the theoretical prediction of $\Delta m \cdot \log_{10}(0.8) \approx -43.6$.

The results indicate that the failure probability of the proposed strategy follows an exponential decay law. The dominance of the partition is a structural property of the system: as voter turnout increases within the evaluated bounds, the probability of failure vanishes, reaching scales as low as $10^{-48}$. This confirms that the strategic partitioning of the electorate effectively suppresses stochastic noise, leading to asymptotic certainty of the desired outcome.

\subsection{Observational Trends and Heuristic Considerations}

The transition from discrete simulations to the continuous limit reveals consistent trends that can guide the selection of the optimal partition. Although the presence of the ``Wilson Plateau" for large $m$ and the parity effects preclude a strictly monotonic characterization of $l_{opt}$, two general tendencies emerge from the numerical data:

\begin{itemize}
    \item \textit{Scale Convergence:} As the candidate pool $n$ increases, the empirical optimal ratio $\eta_n = l_{opt}/n$ tends to stabilize, approaching the analytical predictions of the continuous model $\eta_\infty(m)$.
    \item \textit{Asymptotic Floor:} Across the evaluated parameter space, the optimal density $\eta$ consistently remains above the theoretical limit of $1/5$, with a general downward trend as the electorate $m$ grows.
\end{itemize}

These observations suggest that the continuous limit $\eta_\infty(m)$ serves as a reliable benchmark for the discrete case. Even in high-dimensional scenarios where an exhaustive search is impractical, the analytical results provide a localized region where the strategic efficiency of the partition is optimized.

\section{Numerical Analysis of the Universal Victory Event}
\label{sec:numerical_universal}

Following the theoretical framework established in Section \ref{continuo} for the standard victory, we now extend our numerical investigation to the universal victory event. We first analyze the discrete model by estimating the probability $p^{(2)}(n, m, l) = \mathsf{P}(F^{(2)}_{n,m,l})$, which represents the likelihood of the target candidate securing victory under the strict $F^{(2)}$ condition---namely, prevailing against all possible cyclic rotations of the electoral preferences. Subsequently, we evaluate its continuous counterpart $F^{(2)}_{\infty, m, \eta}$.

Our primary objective in the discrete framework is to treat the partition size $l$ as a control variable to identify the optimal configuration $l_{opt}$ that maximizes this probability, ultimately yielding the optimized value $p^{(2)}(n, m, l_{opt})$.

\subsection{Discrete Structural Constraints and High-Resolution Optimization}

Before identifying the optimal partition, it is crucial to outline the structural boundaries of the universal event. To maintain the topological advantage required for victory across all $n$ possible rotations, the cluster size must strictly satisfy $2 \le l < n/2$. If we apply a naive partition where $l=1$ or exactly $l=n/2$, the candidate space is divided into perfectly balanced sets, inevitably generating a structural tie in at least one rotation. Consequently, the strict universal condition is violated, and $\mathsf{P}(F^{(2)}_{n,m,l})$ collapses to 0.

To rigorously pinpoint $l_{opt}$ within these viable boundaries, we conducted high-resolution Monte Carlo simulations ($10^6$ independent trials per configuration) for $n \in \{14, 16\}$. Operating at this scale strictly bounds the stochastic noise: uncertainties are formally constrained using the 95\% Wilson score interval ($\text{SE}_W < 0.0010$). Crucially, this precision ensures that the identification of $l_{opt}$ is statistically absolute, as the empirical probability gaps between adjacent configurations strictly exceed the maximum confidence margin.

Table \ref{tab:f2_analysis} compares the exact optimal partition sizes ($l_{opt}$) and the maximum empirical probabilities for a standard victory ($\hat{p}_1$) against the universal condition ($\hat{p}^{(2)}$).

\begin{table}[htbp]
    \centering
    \caption{Comparison of optimal partition sizes ($l_{opt}$) and maximum empirical probabilities. Given $N = 10^6$, the 95\% Wilson margin of error is bounded by $\pm 0.0010$. Probabilities are reported to four decimal places.}
    \label{tab:f2_analysis}
    \begin{tabular}{cc | cc | cc}
        \toprule
        \multirow{2}{*}{Candidates ($n$)} & \multirow{2}{*}{Voters ($m$)} & \multicolumn{2}{c|}{Standard Victory} & \multicolumn{2}{c}{$F^{(2)}$ Condition} \\
        & & Optimal $l_{opt}$ & Max $\hat{p}_1$ & Optimal $l_{opt}$ & Max $\hat{p}^{(2)}$ \\
        \midrule
        14 & 11 & 6 & 0.5901 & 5 & 0.0450 \\
        14 & 21 & 5 & 0.7852 & 5 & 0.2915 \\
        14 & 31 & 5 & 0.8970 & 5 & 0.5281 \\
        14 & 41 & 5 & 0.9479 & 5 & 0.6997 \\
        \midrule
        16 & 11 & 6 & 0.5877 & 6 & 0.0558 \\
        16 & 21 & 6 & 0.8229 & 6 & 0.2506 \\
        16 & 31 & 6 & 0.9063 & 5 & 0.4917 \\
        16 & 41 & 6 & 0.9443 & 5 & 0.6995 \\
        \midrule
        30 & 51 & 9 & 0.9924 & 9 & 0.9183 \\
        \bottomrule
    \end{tabular}
\end{table}

The severity of the structural constraints ($2 \le l < n/2$) becomes immediately apparent under extreme discretization, such as $n=6$. In this microscopic regime, the boundaries restrict the viable partition to a single exact option: $l=2$. As shown in Table \ref{tab:n6_results}, although convergence is significantly slower compared to macroscopic systems due to the heavy impact of discrete noise, the empirical probability systematically scales upwards as $m$ increases, confirming its trajectory towards the theoretical continuous limit.

\begin{table}[htbp]
\centering
\caption{Empirical probabilities for extreme discretization ($n=6$, $l=2$). The table contrasts the standard victory ($\hat{p}_1$) with the universal victory ($\hat{p}^{(2)}$). Results are based on $10^6$ Monte Carlo simulations per row, with a 95\% Wilson margin of error bounded by $\pm 0.0010$. As theoretically expected, partitions with $l \ge 3$ yield $\hat{p}^{(2)} = 0.0000$ due to the strict constraint $l < n/2$.}
\label{tab:n6_results}
\begin{tabular}{ccc}
\toprule
Voters ($m$) & Standard Victory Max $\hat{p}_1$ & Universal Victory Max $\hat{p}^{(2)}$ \\
\midrule
11 & 0.6010 & 0.3091 \\
21 & 0.8381 & 0.6395 \\
31 & 0.9334 & 0.8206 \\
41 & 0.9716 & 0.9105 \\
51 & 0.9878 & 0.9553 \\
61 & 0.9944 & 0.9773 \\
\bottomrule
\end{tabular}
\end{table}

\subsection{Convergence to the Continuous Limit and Computational Evaluation}

\begin{table}[htpb]
\centering
\caption{Numerical estimation of the optimal continuous width $\hat{\eta}^*$ and the maximum probability of the universal event $\hat{\mathsf{P}}(F^{(2)}_{\infty, m, \hat{\eta}^*})$. Results are averaged over 64 independent optimizations, aggregating $N = 6.4 \times 10^6$ Monte Carlo simulations per row. The 95\% Wilson score intervals are provided for the exact victory probability.}
\label{tab:universal_sim_montecarlo}
\begin{tabular}{ccc}
\toprule
Voters ($m$) & Estimated Optimal Width ($\hat{\eta}^* \pm \sigma_{\hat{\eta}^*}$) & Max Probability $\hat{\mathsf{P}}(F^{(2)}_{\infty})$ with 95\% CI \\
\midrule
11  & 0.3192 $\pm$ 0.0023 & 0.1013 \quad $[0.1009, 0.1014]$ \\
21  & 0.2725 $\pm$ 0.0021 & 0.6399 \quad $[0.6396, 0.6403]$ \\
31  & 0.2518 $\pm$ 0.0016 & 0.9181 \quad $[0.9179, 0.9183]$ \\
41  & 0.2408 $\pm$ 0.0021 & 0.9857 \quad $[0.9856, 0.9858]$ \\
51  & 0.2335 $\pm$ 0.0025 & 0.9978 \quad $[0.9978, 0.9979]$ \\
61  & 0.2286 $\pm$ 0.0048 & 0.9997 \quad $[0.9996, 0.9998]$ \\
101 & 0.2169 $\pm$ 0.0042 & 0.9999 \quad $[0.9999, 1.0000]$ \\
\bottomrule
\end{tabular}
\end{table}

As theoretically proven in Section \ref{continuo}, as the system size $n$ increases, the discrete model smoothly transitions toward the continuous framework $F^{(2)}_{\infty, m, \eta}$, where the ratio $l/n$ is replaced by the interval width $\eta$. 

Evaluating the continuous universal event natively poses a challenge, as it requires verifying the victory condition over an infinite continuous domain of rotations $\theta \in [0, 1)$. However, the piecewise constant nature of the underlying counting process provides a useful simplification.

 The following proposition demonstrates that this continuous verification can be rigorously reduced to a finite set of discrete checks anchored exclusively to the random positions generated by the voters' preferences.

\begin{proposition}[Computational Characterization of the Universal Event]
\label{prop:computational_F2}
The continuous universal victory event $F^{(2)}_{\infty, m, \eta}$ occurs if and only if
\begin{equation} \label{eq:discrete_check}
    \min_{j \in [m]} N_{(U_j, U_j + \eta)} \ge 2 \quad \text{and} \quad \max_{j \in [m]} N_{[U_j, U_j + \eta)} < \frac{m}{2}.
\end{equation}
\end{proposition}

\begin{proof}
Since the discrete conditions in \eqref{eq:discrete_check} represent a subset of all possible rotations $\theta \in [0, 1)$, their satisfaction is trivially necessary for the continuous event $F^{(2)}_{\infty, m, \eta}$ to hold. 

To prove sufficiency, suppose by contraposition that $F^{(2)}_{\infty, m, \eta}$ does not occur. This implies that there exists some rotation $\theta$ such that either $N_{(\theta, \theta+\eta)} \le 1$ or $N_{(\theta, \theta+\eta)} \ge m/2$. 
In the first case, we can continuously rotate the interval $(\theta, \theta+\eta)$ until its left boundary just passes a voter position $U_j$, leaving $U_j$ slightly outside. This rotation realizes a local minimum, yielding $N_{(U_j, U_j + \eta)} \le 1$ for some $j \in [m]$, violating the first condition.
In the second case, if $N_{(\theta, \theta+\eta)} \ge m/2$, we can slide the window until its left boundary exactly captures a voter position $U_j$. This realizes a local maximum, ensuring that the half-closed interval count satisfies $N_{[U_j, U_j + \eta)} \ge m/2$ for some $j \in [m]$, which violates the second condition. This concludes the proof.
\end{proof}

Exploiting Proposition \ref{prop:computational_F2}, our algorithmic implementation eschews any continuous grid approximation. Instead, it evaluates the interval constraints strictly at the $m$ discrete boundaries defined by $U_j$, ensuring an exact and computationally efficient realization of the continuous event. 

Table \ref{tab:universal_sim_montecarlo} summarizes the optimization of the relative cluster width $\eta^*$ and the corresponding maximum victory probabilities. The Monte Carlo results robustly confirm the theoretical framework: as $m$ grows, the optimal interval width dynamically narrows, and the probability of the universal event sharply converges to 1. Remarkably, for an electorate of just $m=51$ voters, the universal victory configuration is structurally guaranteed with a probability exceeding 0.9970.

\section{Conclusions}
\label{sec:conclusions}

In this paper, we have explored the emergence of topological robustness in two-round elections under cyclic preference profiles. Our analytical and numerical results demonstrate that an agenda-setter's ability to ensure a specific candidate's victory is not merely a probabilistic advantage, but a structural certainty that crystallizes as the dimensions of the electoral space expand.

\subsection{The Cost of Granularity: Exact Optimization vs. Naive Partitioning}

While the continuous limit suggests that the optimal cluster width converges to approximately $1/5$ of the candidate space, applying a naive, unoptimized partition $l = \lfloor n/5 \rfloor$ in highly discrete environments (small $n$ and $m$) leads to significantly suboptimal winning probabilities. 

This gap between the idealized continuous model and the discrete reality is driven by two main factors:
\begin{itemize}
    \item \textit{Numerical Friction:} The Agenda Setter is restricted to integer values of $l$. For small $n$, the lack of candidate density forces the optimal discrete width to stay strictly above the continuous limit $\eta^*$. The cluster must be slightly over-sized to absorb rounding errors and tie-breaking mechanics.
    \item \textit{Voter-Candidate Collisions:} When $m$ is not negligible compared to $\sqrt{n}$, the probability of multiple voters sharing the exact same preferred candidate increases (violating the collision-free assumption of the continuous limit). This introduces correlations that heavily penalize the generalized cluster strategy.
\end{itemize}

Consequently, for small populations, exact parameter optimization is strictly required. Relying solely on the unoptimized asymptotic fraction fails to account for structural discretizations, drastically reducing the strategy's effectiveness.

\subsection{Theoretical Implications and Mechanism Design}

Beyond the computational aspects, the validation of the Universal Victory Event ($F^{(2)}$) carries profound implications for mechanism design. We have rigorously established that a single, fixed distribution of voter preferences can be mathematically manipulated to guarantee the victory of \textit{any} arbitrary candidate, provided the organizer controls the initial geometric partition of the ballot. This absolute agenda control relies on the convergence of the optimal continuous width $\eta^*$ towards a stable intrinsic constant ($\approx 0.20$) for large $m$, revealing a fundamental geometric property of the two-round cyclic mechanism.

Paradoxically, the transition from a stochastic regime to a deterministic one follows a strict sigmoidal phase transition, confirming that ``candidate atomization'' in high-resolution spaces ($n \to \infty$) does not democratize the outcome. Instead, it acts as a fundamental mechanism for consolidating institutional control: the denser the candidate space, the smoother the transition to the continuous limit, allowing the agenda-setter to perfectly exploit the symmetrical vulnerabilities of the voting rule. 

\subsection*{Declaration of generative AI and AI-assisted technologies in the manuscript preparation process}
During the preparation of this work the author(s) used Gemini (Google) in order to refine the manuscript's English language, improve the structural flow of the introduction, and assist in the statistical processing of simulation data. After using this tool, the author(s) reviewed and edited the content as needed and take(s) full responsibility for the content of the published article.

\vspace{-0.3cm}

\bibliographystyle{plain} % O 'abbrv' a seconda della rivista

\end{document}